\documentclass[conference,letterpaper]{IEEEtran}
\usepackage{multirow}
\usepackage[framed,autolinebreaks,useliterate]{mcode}%numbered
\usepackage{mathtools}
\usepackage{nicematrix}
\usepackage{float}
\usepackage{makecell}
\usepackage{tabularx}
\usepackage{subcaption}
\usepackage{booktabs}
\usepackage{cancel}
\usepackage{enumitem} % to use (a) instead of 1)
\usepackage{ytableau}
\usepackage{graphicx} %package to manage images
\graphicspath{ {images/} }

\usepackage{color, colortbl}
\usepackage{multirow}

\usepackage{cuted}
\usepackage[linesnumbered,ruled,vlined]{algorithm2e} % ,noend %remove "vlined" for puting "end"

\usepackage[utf8]{inputenc}
\usepackage[english]{babel}
\usepackage{amsthm,amssymb}
\usepackage{amsmath}
\usepackage{bm}
\usepackage{optidef}
\usepackage{mathrsfs}

\usepackage[font=small,labelfont=bf,labelsep=space]{caption}
\usepackage{xcolor}
\def\BibTeX{{\rm B\kern-.05em{\sc i\kern-.025em b}\kern-.08em
    T\kern-.1667em\lower.7ex\hbox{E}\kern-.125emX}}

\usepackage[style=ieee,backend=biber]{biblatex}
\AtBeginBibliography{\footnotesize} %scriptsize

\usepackage{tikz}
\usetikzlibrary{matrix,backgrounds,positioning}

\usetikzlibrary{calc,patterns,decorations.pathreplacing,decorations.pathmorphing,matrix,positioning,arrows.meta,fit,bending}
\usepackage{pgfplots}
\pgfplotsset{compat=1.16}

\theoremstyle{definition}
\newtheorem{example}{Example}
\newtheorem{theorem}{Theorem}
\newtheorem{proposition}{Proposition}
\newtheorem{lemma}{Lemma}
\newtheorem{corollary}{Corollary}

\newtheorem{definition}{Definition}

\newtheorem{remark}{Remark}

\DeclareMathOperator{\bin}{bin}
\DeclareMathOperator{\supp}{supp}

\DeclareMathOperator{\w}{w}
\DeclareMathOperator{\dist}{d}

\newcommand{\rank}{\mathrm{rank}}
\newcommand{\wm}{\w_{\min}}

\newcommand{\A}{\mathcal{A}}

\newcommand{\I}{\mathcal{I}}

\newcommand{\F}{\mathcal{F}}

\newcommand{\C}{\mathcal{C}}

\newcommand{\R}{\mathcal{R}}

\newcommand{\bA}{\mathbf{A}}
\newcommand{\bB}{\mathbf{B}}
\newcommand{\bH}{\mathbf{H}}
\newcommand{\bM}{\mathbf{M}}
\newcommand{\bId}{\mathbf{Id}}
\newcommand{\bQ}{\mathbf{Q}}

\newcommand{\bve}{\bm{\varepsilon}}

\newcommand{\bp}{\mathbf{p}}

\newcommand{\bP}{\mathbf{P}}
\newcommand{\bGN}{\boldsymbol{G}_N}
\newcommand{\bc}{\boldsymbol{c}}
\newcommand{\bu}{\boldsymbol{u}}
\newcommand{\bx}{\boldsymbol{x}}

\newcommand{\bv}{\boldsymbol{v}}
\newcommand{\bw}{\boldsymbol{w}}

\newcommand{\bG}{\boldsymbol{G}}
\newcommand{\ind}{\operatorname{ind}}
\newcommand{\terms}{\operatorname{terms}}
\newcommand{\ev}{\operatorname{ev}}

\newcommand{\ff}{\mathbb{F}}
\newcommand{\ft}{\mathbb{F}_2}

\newcommand{\Alow}{{\rm LTA}(m,2)}

\newcommand{\GL}{{\rm GL}(m,2)}
\newcommand{\Mon}{\mathcal{M}_{m}}
\newcommand{\weako}{\preceq_{\mathbf{w}}}
\newcommand{\Rm}{{\mathbf {R}}_m}

\DeclareFontFamily{U}{mathx}{}
\DeclareFontShape{U}{mathx}{m}{n}{<-> mathx10}{}
\DeclareSymbolFont{mathx}{U}{mathx}{m}{n}
\DeclareMathAccent{\widecheck}{0}{mathx}{"71}
\definecolor{highlightcolor}{rgb}{1.0, 0.8, 0.8} % Light red (red!20)

\IEEEoverridecommandlockouts                              % This command is only
\title{
 Algebraic Properties of PAC Codes
}

\author{

\IEEEauthorblockN{Vlad-Florin Dr\u{a}goi}
\IEEEauthorblockA{%Faculty of Exact Sciences,\\
Aurel Vlaicu University\\
Arad, Romania\\
Email: vlad.dragoi@uav.ro}
\and
\IEEEauthorblockN{Mohammad Rowshan}
\IEEEauthorblockA{%School of Electrical Engineering and Telecommunications,\\
University of New South Wales\\ 
Sydney, Australia\\ 
Email: m.rowshan@unsw.edu.au}
}

\begin{document}

\maketitle
\pagestyle{empty}

%%%%%%%%%%%%%%%%%%%%%%%%%%%%%%%%%%%%%%%%%%%%%%%%%%%%%%%%%%%%%%%%%%%%%%%%%%%%%%%%
\begin{abstract}
We analyze polarization-adjusted convolutional codes using the algebraic representation of polar and Reed-Muller codes. We define a large class of codes, called generalized polynomial polar codes which include PAC codes and Reverse PAC codes. We derive structural properties of generalized polynomial polar codes, such as duality, minimum distance. We also deduce some structural limits in terms of number of minimum weight codewords, and dimension of monomial sub-code.   
\end{abstract}
% keywords are not needed for ISIT according to the tempalte.
\begin{IEEEkeywords}
Polynonomial codes, polar codes, PAC codes, monomial subcode, dual code.
\end{IEEEkeywords}

%\cite{arikan}

\section{Introduction}
Polar codes %, identified as a code family capable of achieving the capacity of any binary-input discrete memoryless channel (BI-DMC)
\cite{arikan} are members of a larger family of codes called decreasing monomial codes (DMC) \cite{bardet2016crypt}, which also includes Reed--Muller codes. The algebraic formalism in \cite{bardet} revealed that  dual of a DMC is still a DMC, the minimum distance of DMC is directly related to its maximum degree monomials, and the permutation group of DMCs contains the lower-triangular affine group (LTA). These properties have since been exploited in code construction \cite{mondelli-construction,Wang-Dragoi2023}, cryptography \cite{bardet2016crypt,dragoi-2019}, closed-form and algorithmic weight enumeration \cite{vlad1.5d,rowshan2024weight,yao,liu2024method}, and automorphism-ensemble decoding \cite{PBL21}. However, this algebraic framework presently covers only decreasing monomial codes; more general monomial or polynomial-based families lie outside its scope, despite their increasing relevance. The arrival of polarization-adjusted convolutional (PAC) codes \cite{arikan2}, reverse PAC codes \cite{gu2024reverse}, and profile-shifted PAC codes \cite{gu2025pac} calls for a more refined understanding of the structural properties of these families. 

PAC codes enhance the finite-length performance of polar codes by concatenating a one-to-one convolutional pre-transform with the polar transform, thereby preserving channel polarization while effectively reducing the number of low-weight codewords, compared to plain polar codes, and improving error-correction performance at short block lengths, under near-ML decoding (e.g., sequential or list decoding), close to theoretical limits. %PAC codes were originally proposed to avoid the capacity loss inherent in fixing bit-channel inputs of conventional polar codes by using a convolutional precoding stage ahead of the polar encoder, which effectively reduces the number of low-weight codewords and improves minimum distance behavior under near-ML decoding (e.g., sequential or list decoding) compared to pure polar constructions. Structurally, PAC codes can be interpreted as equivalent to concatenations of outer rate-1  convolutional codes and inner polar- or Reed–Muller–like codes, leveraging algebraic code properties to further enhance distance properties \cite{Arikan2020_PAC_Systematic,Moradi2024_PAC_Cyclic}, although  algebraic characterization of their dual codes or full permutation group has not been fully established in the literature; existing work primarily focuses on approximate weight distributions, distance spectra as well as decoding performance \cite{rowshan-pac1,rowshan-precoding,rowshan22EnuPAC} rather than explicit dual or automorphism group descriptions. 
Yet, existing work on PAC and reverse PAC codes has mainly focused on performance and approximate distance spectra \cite{rowshan-pac1,rowshan-precoding,rowshan22EnuPAC}, or leveraging algebraic code properties to further enhance distance properties \cite{Arikan2020_PAC_Systematic,gu2023rate,Moradi2024_PAC_Cyclic}, and does not provide a systematic algebraic description of their duals or permutation groups. In particular, PAC codes are not decreasing monomial codes, are not closed under LTA permutations, and their codewords are not evaluations of polynomials in the span of the plain-polar information set~$\mathcal{A}$.

%  We are also motivated by some incipient remarks that one can infer on PAC codes. For example, PAC codes are not decreasing monomial codes, their generator matrices are different and not isomorphic to polar codes generator matrix; PAC code are not closed under LTA permutations; codewords of PAC codes do not correspond to evaluations of polynomials in the span of $\A$, the information set of plain-polar codes. In this article PAC codes are described as a particular case of a more generic family of codes, \emph{generalized polynomial polar codes}. The key idea is to use the algebraic description of polar and Reed-Muller codes in order to redefine PAC codes, which are upper polynomial polar codes, in our terms. We reveal structural properties of PAC codes.
% First,  we show that upper polynomial polar codes admit a large decreasing monomial sub-code, as well as a large monomial sub-code, except for very low/high rate codes. Second, we analyze the dual of polynomial polar codes. We begin by providing a set of properties of the dual of monomial codes. These are mainly based on a better understanding of the algebraic properties of the plain polar kernel matrix. We demonstrate that the dual of an upper polynomial polar code is a lower polynomial polar code.
% Thirdly, we derive lower bounds on the number of minimum weight codewords of upper polynomial polar codes. We demonstrate that a PAC code preserves the minimum distance of the plain polar code. 

PAC codes exhibit several structural differences from plain polar codes: they are not decreasing monomial codes, their generator matrices are not isomorphic to the polar generator, they are not closed under LTA permutations, and their codewords are not evaluations of polynomials in the span of the polar information set~$\mathcal{A}$. In this work we place PAC codes inside a broader class of \emph{generalized polynomial polar codes} and reinterpret them as \emph{upper polynomial polar codes} using the algebraic representation of polar and Reed–Muller codes. Within this framework we show that, except at very low or very high rates, upper polynomial polar codes contain large decreasing-monomial and monomial subcodes; we characterize duals by first analyzing duals of monomial codes and then proving that the dual of an upper polynomial polar code is a lower polynomial polar code; and we derive lower bounds on the number of minimum-weight codewords, from which we conclude that a PAC code preserves the minimum distance of the underlying plain polar code.

%%%%%%%%%%%%%%%%%%%%%%%%%%%%%%%%%%%%%%%%%%%%%%%%%%%%

\section{Decreasing monomial codes}%
We denote by $\ft$ the finite field with two elements. 
Also, subsets of consecutive integers are denoted by $[\ell,u]\triangleq\{\ell,\ell+1,\ldots,u\}$. 
The \emph{support} of a vector $\bc = [c_0,\ldots,c_{N-1}] \in \ft^N$ is defined as $\supp(\bc) \triangleq \{i \in [0,N-1] \mid c_i \neq 0\}$. 
The cardinality of a set is denoted by $|\cdot|$ and the set difference by $\backslash$. 
The Hamming \emph{weight} of a vector $\bc \in \ft^N$ is $\w(\bc)\triangleq |\supp(\bc)|$. 
The binary vectors $\bin(i)=(i_0,\dots,i_{m-1})\in \ft^m$, where $i=\sum_{j=0}^{m-1} i_j2^j$, are written so that $i_{m-1}$ is the most significant bit. In addition, we endow $\ft^m$ with the decreasing index order relation. The scalar product between two $N$-length vectors will be denoted $\bv\cdot \bw=\sum_{i=0}^{N-1}v_iw_i.$ Also, we will denote the Schur product (or component-wise product) of $\bv$ and $\bw$ by 
$\bv\odot\bw=(v_0w_0,\dots,v_{N-1}w_{N-1}).$ %Hence, the scalar product of two vectors if the sum of the components of their Schur product. 

A $K$-dimensional subspace $\C$ of $\ft^N$ is called a linear $(N,K,d)$ \emph{code} over $\ft$ 
where $d$ is the minimum distance of code $\C$; that is,
%\begin{equation*}
$\dist(\C) \triangleq \min_{\bc,\bc' \in \C, \bc \neq \bc'} \dist(\bc,\bc')=\wm$, which is equal to minimum weight of codewords excluding all-zero codeword.
%\end{equation*}
Let us collect all $\w$-weight codewords of $\C$ in set $W_{\w}$ as
$
    W_{\w}(\C) = \{ \bc\in\C \mid \w(\bc)=\w \}.
$
A \emph{generator matrix} $\bG$ of an $(N,K)$-code $\C$ is a $K \times N$ matrix in $\ft^{K\times N}$ whose rows are $\ft$-linearly independent codewords of $\C$. Then $\C = \{\bv \bG \colon \bv \in \ft^K\}$. We will denote matrices and vectors using bold symbols. The transpose of $\bG$ is denoted by $\bG^t.$

\paragraph{Multivariate monomial and polynomials}
Let $m$ be a fixed integer that represents the number of different variables $\mathbf{x}\triangleq(x_{0},\dots,x_{m-1})$ and $\ft[x_0,\dots{},x_{m-1}]$ be the set of polynomials in $m$ variables. %Since we are dealing with binary codes, we will identify $x_i$ with $x_i^2$ (using the Frobenius endomorphism) and 
Also, consider the ring $\Rm=\ft[x_0,\dots{},x_{m-1}]/(x_0^2-x_0,\dots,x_{m-1}^2-x_{m-1}).$%, where ''$\cdot$" denotes multiplication between polynomials. We use the same symbol as for scalar product, however, it can be understood from the context when $\cdot$ means polynomial multiplication or scalar product. 

A particular basis for $\Rm$ is the monomial basis, where any $m$-variate monomial can be defined as $\mathbf{x}^{\bin(i)}=x_0^{i_0}\cdots x_{m-1}^{i_{m-1}},$
%$$\mathbf{x}^{\bin(i)}=\prod_{j=0}^{m-1}x_{j}^{i_{j}}=x_{0}^{i_{0}} \cdots x_{m-1}^{i_{m-1}},$$ 
where $\bin(i)=(i_{0},\dots,i_{m-1})$ with $i_j\in\{0,1\}.$ 
Denote the set of all monomials in $\Rm$ by 
$
\Mon\!\triangleq\!\left\{\mathbf{x}^{\bin(i)} \mid\bin(i) \in \mathbb{F}_{2}^{m}\right\}
.$

%For example, for $m=2$ we have $\Mon=\{x_0x_1,x_1,x_0,\bm{1}\}.$ 
Furthermore, the \emph{support of monomial} $f=x_{l_1}\dots x_{l_s}$ is $\ind(f)=\{l_1,\dots,l_s\}$ and its degree is $\deg(f)=|\ind(f)|.$ The degree induces a ranking on any monomial set $\I\subseteq\Mon$, that is, $\I=\bigcup_{j=0}^{m}\I_j$, where $\I_j=\{f\in\I\mid\deg(f)=j\}.$ Any polynomial will be written as the sum of its terms, i.e., $\terms(P)\subset\Mon, P=\sum_{f\in\terms(P)}f.$  

\paragraph{Monomial codes}
Define the evaluation function by %of a polynomial, $\ev(P)$, by 
\begin{equation}
\begin{array}[h]{ccccc}
\Rm    & \to &\ft^{2^m}\\
P& \mapsto &\ev(P) = \big(g(\bin(i)) \big)_{\bin(i) \in \ft^m}
\end{array}
\end{equation}

%The function $\ev(\cdot)$ defines a vector space isomorphism between the vector space $(\Rm,+,\cdot)$ and $(\ft^{2^m},+,\cdot).$

%We can now define monomial/polynomial codes. 
\begin{definition}
Let $\I\subseteq\Mon.$ A monomial code is the vector subspace $\C(\I) \subseteq \ft^{2^m}$ 
generated by $\{ \ev(f) ~|~ f \in \I\}$.\\
When $\I\subset \Rm$ is a set of polynomials we call $\C(\I)$ a polynomial code.  
\end{definition}

%To demonstrate onstructing monomial codes, we
Let $\bG_{2^m} \triangleq \begin{pmatrix} 1 & 0 \\ 1 & 1 \end{pmatrix}^{\otimes m}$, where $\otimes m$ denotes the Kronecker power. 
In \cite{dragoi17thesis}, it was shown that $\bG_{2^m}$ is a basis for the vector space $\ft^{2^m}$, based on the mapping% This fact is grounded in the~following:
\begin{equation}\label{eq:mapping}
\begin{array}[h]{ccccc}
[0,2^m-1]    & \to &\Mon   & \to &\ft^{2^m}\\
i& \mapsto &g\triangleq\mathbf{x}^{{\bin(2^m-1-i)}} & \mapsto &\bG_{2^m}[i]=\ev(g) 
\end{array}
\end{equation}

We call $i$ the index of $g=\mathbf{x}^{\bin(2^m-1-i)}$, e.g., for $m=3$ the index of $x_0x_2$ is $2$ since $\bin(7-2)=(1,0,1)$ and $\mathbf{x}^{(1,0,1)}=x_0^1x_1^0x_2^1=x_0x_2.$

The monomial code $\C(\I)$ is the code spanned by the matrix $\bG_{\{i_g\mid g\in \I\}}.$

{
\begin{center}
\setlength{\tabcolsep}{2pt}
%\footnotesize
\begin{tabular}{cccc|c|ccccc} 
 $i$&$\bin(i)$& $\bin(3-i)$ & $g$&&$\mathbf{z}=(z_0\;z_1):$ & 11 & 01 & 10 & 00 \\
\hline\hline 
0&$(0,0)$&$(1,1)$& $x_0x_1$& & $\mathrm{ev}\left(x_0 x_1\right)$ & 1 & 0 & 0 & 0 \\
1&$(1,0)$&$(0,1)$& $x_1$&  & $\mathrm{ev}\left(x_1\right)$ & 1 & 1 & 0 & 0 \\
2&$(0,1)$&$(1,0)$& $x_0$&  & $\mathrm{ev}\left(x_0\right)$ & 1 & 0 & 1 & 0 \\
3&$(1,1)$&$(0,0)$& $\bm{1}$&  & $\operatorname{ev}(\bm{1})$ & 1 & 1 & 1 & 1
\end{tabular}
\end{center}
}

Selecting rows of $\bGN$ is equivalent to selecting monomials from $\Mon.$ We use $\A$ to denote a set of indices and $\C_{\bGN}(\A)$ to denote the linear code generated by the set of rows of $\bGN$ indexed by $\A$, where $\A \subseteq [0,N-1]$. A polar code of length $N=2^n$ is constructed by selecting a set $\A$ of indices $i\in[0,N-1]$ corresponding to high reliability synthetic channels \cite{arikan}. $\A$ is used for information bits, while the rest of the synthetic channels (indices in $\A^c \triangleq [0,N-1]\setminus \A$) are used to transmit a known value, '0' by default, which are called \emph{frozen bits}. 
%
% Polar codes with the information set $\A$ can also be considered as monomial codes \cite{dragoi17thesis} where the relation between the generating set of monomials $\I\subset\Mon$ and the information set $\A$ is as follows: 
Any set of indices $\A$ maps into a monomial set $\I$ using \eqref{eq:mapping}, hence, we have $\C(\I)=\C_{\bG}(\A).$
% as follows
{
\begin{align}
    \label{eq:i_f}
    \A &= \left\{\sum_{{\qquad j\in \ind(\widecheck{f})}}2^j  \;\mid \; f \in \I\right\} \\
\I &= \left\{\prod_{{\qquad\qquad j \in \supp(\bin(2^m-1-i))}} x_j  \;\mid\; i \in \A\right\}.
\end{align}

}

\paragraph{Decreasing monomial codes}\label{par:dec_monomial_codes}

Let ``$|$" denote the divisibility between monomials, i.e., $f|g$ iff $\ind(f)\subseteq\ind(g).$ %Also, the greatest common divisor of two monomials is $\gcd(f,g)=h$ with $\ind(h)=\ind(f)\cap\ind(g).$

\begin{definition}\label{def:order}Let $m$ be a positive integer and $f,g\in\Mon.$ Then $f\weako g$ if and only if $ f|g.$ When $\deg(f)=\deg(g)=s$ we say that $f\preceq_{\mathbf{sh}} g$ if $\forall\;1\leq\ell\leq s\;\text{ we have }\;  i_\ell \le j_\ell$, where $f=x_{i_1}\dots x_{i_s}$, $g=x_{j_1}\dots x_{j_s}$. 
Define $f\preceq g\quad \text{iff}\quad \exists g^*\in \Mon\;\text{s.t.}\; f\preceq_{\mathbf{sh}} g^*\weako g$.
\end{definition}

\begin{definition}\label{def:dec_set}
      A set $\I \subseteq \Mon$ is \emph{decreasing}  if and only if ($f \in \I$ and $g \preceq f $) implies $g \in \I$. 
      A decreasing interval $[1,f]_{\preceq}$ is a decreasing monomial set of $\{g\in\Mon\mid g \preceq f\}.$
\end{definition}

Any monomial code $\C(\I)$ with $\I$ decreasing is called \emph{decreasing monomial code}. Both Polar and Reed-Muller codes are decreasing monomial codes \cite{bardet}.

\paragraph{Permutation group}\label{ssec:perm_grp}
The set of permutations that globally leave the code invariant, forms the \emph{automorphism group} of the code $\C.$ 
A bijective affine transformation over $\ft^m$ is represented by a pair $(\bB, \bm{\varepsilon})$ where $\bB=(b_{i,j})\in \ff_2^{m\times m}$ is an invertible matrix lying in the general linear group $\GL$ and $\bve$ in $\ft^m$. The action of $(\bB, \bve)$ on a monomial $g=\prod_{i\in\ind(g)}x_i$ (denoted by $(\bB, \bve)\cdot g$) replaces each variable $x_i$ of $g$ by a variable $y_i$ as
$y_i = x_i + \sum_{j=0}^{i-1} b_{i,j} x_j + \varepsilon_i$. This new variable $y_i$, is in fact a linear form (a polynomial in which all terms have a degree at most 1), where the maximum variable is $x_{i}.$%, w.r.t. the order relation $\preceq.$

In \cite{bardet} it was proved that the lower triangular affine transformation denoted by $\Alow$ is a subgroup of the permutation group of any decreasing monomial code. $\Alow$ is a subgroup of the affine group, where $\bB\in\GL$ is a lower triangular binary matrix with $b_{i,i}=1$ and $b_{i,j}=0$ whenever $j>i$. Therefore, the action of $\Alow$ can be expressed as the following mapping from $\ft^m$ to itself: 
$
\bx \rightarrow \bB \bx + \bve.
$
The set of polynomials resulting from the action of $\Alow$ on a monomial is collected in a set named \emph{orbit} and denoted by
$\Alow \cdot f = \{(\bB,\bve) \cdot f\mid (\bB,\bve) \in \Alow\}.$
    
Since $\Alow$ acts as a permutation on $\ev(f)$, all elements in $\Alow\cdot f$ have the same Hamming weight. %This 

\paragraph{Minimum Weight Codewords}\label{ssec:wmin_orbit}
Minimum weight codewords are counted using a particular subgroup of $\Alow$.

\begin{definition}[\cite{bardet,dragoi17thesis}] For any $g\in\Mon$ define $\Alow_g$ as the subgroup $\Alow$ by "collecting actions
 of $(\bB,\bve)$ that satisfy" 
 \vspace{-5pt}
\[
 \varepsilon_i = 0 \text{ if } i \not \in \ind(g) 
 ~~~\text{ and }~~ 
 b_{ij} = 
 \left \{ 
 	\begin{array}{l}
		0 \text{ if } i \not \in  \ind(g), \\
		0 \text{ if }  j \in \ind(g).
	 \end{array}
 \right.
\] 
\end{definition}

In \cite{bardet} it was demonstrated that $\Alow\cdot f=\Alow_f\cdot f$, hence, $\Alow_f\cdot f$ being a finer estimation of the orbit. This group can be written as a semi-direct product between the group of invertible matrices having one on the diagonal and the group of translations. A variable $x_i$ can be translated by a scalar $\varepsilon_i\in\ft$ onto $x_i+\varepsilon_i$. Hence, $f\in\Mon$ admits $2^{\deg(f)}$ translations. 
Regarding linear mapping, $x_i$ can be transformed into a "new variable" ($y_i$), where $y_i=x_i+\sum_{j=0, j \notin \operatorname{ind}(f)}^{i-1} b_{i,j} x_{j}$. The extra variables considered in $y_i$ express the degree of freedom we have on $x_i$ and are denoted by $\lambda_f(x_i)=|\{j\in[0,i) \mid j\notin\ind(f)\}|.$ To each monomial $f=x_{i_{0}}\dots x_{i_{s-1}}$ with $i_0<\dots< i_s$ one can associate a partition of length $\deg(f)$ defined by $\lambda_f=(i_{s-1}-(s-1),\dots,i_0-0)$ (see \cite{dragoi17thesis}). The total number of free variables on all $x_i$ in the support of $f$ is $|\lambda_f(f)|=\sum_{i\in\ind(f)}\lambda_f(x_i)$, (or simply $\lambda_f$) from which 
    $2^{|\lambda_f|}$ (since we are defined over $\ff_2$) possible actions on $f.$ 

    The definition can be extended to any monomial $g=x_{j_{0}}\dots x_{j_{l-1}}$ satisfying $g|f$, that is, $\lambda_f(g)$ is the partition of length $l$ defined by $\lambda_f(g)=(\lambda_f(x_i))_{i\in\ind(g)}$, which yields $ |\lambda_f(g)|=\sum_{i\in\ind(g)}\lambda_f(x_i).$ 

 Finally, we get the well known formula from \cite{bardet}% and the enumeration of minimum weight codewords of a de creasing monomial code.
\begin{equation}\label{eq:A_wm}
    \left|\Alow_f \cdot f\right| = 2^{\deg(f)+|\lambda_f|},
\end{equation}
\begin{equation}\label{eq:sum_A_wm}
    |W_{2^{m-r}}(\I)|=\sum\limits_{f\in \I_r}2^{\deg(f)+|\lambda_f|}.
\end{equation}

%%%%%%%%%%%%%%%%%%%%%%%%%%%%%%%%%%%%%%%%%%%%%%%%%%%
%%%%%%%%%%%%%%%%%%%%%%%%%%%%%%%%%%%%%%%%%%%%%%%%%%%
\section{Generalized polynomial polar codes}
Let $\C_{\bGN}(\A)$ be a polar code with information set $\A$ with length $N=2^m$. The encoding step of a pure polar code is defined by the vector-matrix multiplication $\bv\bGN$ where $\bv$ is the rate-profile dependent vector, meaning that $v_i=0,\forall i\in \A^c.$ A PAC code applies a pre-transformation via a binary polynomial $\bp(x)$, effectively changing the input vector $\bv$ to $\bu=\bv \star \bp \bmod 2$, so that the codeword is $\bc=\bu \bGN=\bv \bP \bGN$, where $\bP\in \mathbb{F}_2^{N \times N}$ is the Toeplitz matrix of $\bp$, denoting exactly the rate-1 convolution operation 
\[\bP=\left(\begin{array}{ccccccc}
     p_0& p_1 & \dots & p_l & 0 &\dots &0 \\
     0 & p_0 & \dots &p_{l-1}& p_l & \dots &0\\
     \vdots &\ddots &\ddots & \ddots &\ddots &\ddots &\vdots\\
    % 0 & 0 & 0 & 0 & 0 & p_0 &p_1\\
     0 & 0 & 0 & 0 & 0 & 0 &p_0
\end{array}\right),\] with $p_0=1.$ Its inverse
$
\bH=\bP^{-1}
$
implements the terminated deconvolution:
$
\bv=\bH \bu.
$

Here we will consider a broader class of polynomial codes, that we call \emph{Generalized polynomial polar codes}. 
\subsection{Definitions and Properties}
\begin{definition}
    Let $\C_{\bGN}(\A)$ be a polar code with information set $\A$ and let $\bP\in \ft^{N\times N}$ be an invertible matrix. The code $\C_{\bP\bGN}(\A)$ is called a Generalized polynomial polar code. 

    When $\bP$ is lower/upper triangular with  $P_{i,i}=1$ the code $\C_{\bP\bGN}(\A)$ is called lower/upper polynomial polar code.
\end{definition}

\begin{remark}
\begin{enumerate}
    \item A generalized polynomial polar code is by default a polynomial code. Rare cases, e.g., extremely low/high rate polar codes could generate Generalized polynomial polar codes that are monomial codes. Also, particular pairs $(\bP,K/N)$ generate monomial codes. 
    \item 
    PAC codes are a particular sub-class of generalized polynomial polar codes, where $\bP$ is the upper triangular Toeplitz matrix of $\bp.$  
    \item Generalized polynomial codes admit an inverse matrix $\bH=\bP^{-1}$. Identically, we can compute $\bv=\bH\bu.$
\end{enumerate}
\end{remark}

\begin{table*}[!h]
    \centering
    \resizebox{0.98\textwidth}{!}{
    \begin{tabular}{|c|c|c||c|c|c|}
\toprule
\multicolumn{3}{|c|}{Indices, monomials and polynomials}&\multicolumn{3}{|c|}{Their complements}\\
\midrule
\textbf{Row-Index $i$} & \makecell{\textbf{Monomial $f$ s.t. $\bG[i]=\ev(f)$ }} & \makecell{\textbf{Polynomial $P$ s.t. $\bP\bGN[i]=\ev(P)$}} & \textbf{$2^m-1-i$} & \makecell{\textbf{ $\widecheck{f}$}} & \makecell{\textbf{ $\widecheck{P}$}} \\
\midrule

\rowcolor{highlightcolor}0 & $x_0x_1x_2x_3x_4x_5$
& $x_0x_1x_2x_3x_4x_5 + x_0x_2x_3x_4x_5 + x_2x_3x_4x_5$
& 63 & $1$
& $x_0x_1 + x_1 + 1$ \\ \hline

\rowcolor{highlightcolor}1 & $x_1x_2x_3x_4x_5$
& $x_0x_1x_3x_4x_5 + x_1x_2x_3x_4x_5 + x_2x_3x_4x_5$
& 62 & $x_0$
& $x_0x_1 + x_0 + x_2$ \\ \hline

\rowcolor{highlightcolor}2 & $x_0x_2x_3x_4x_5$
& $x_0x_1x_3x_4x_5 + x_0x_2x_3x_4x_5 + x_1x_3x_4x_5$
& 61 & $x_1$
& $x_0x_2 + x_1 + x_2$ \\ \hline

\rowcolor{highlightcolor}3 & $x_2x_3x_4x_5$
& $x_0x_3x_4x_5 + x_1x_3x_4x_5 + x_2x_3x_4x_5$
& 60 & $x_0x_1$
& $x_0x_1 + x_0x_2 + x_1x_2$ \\ \hline

\rowcolor{highlightcolor}4 & $x_0x_1x_3x_4x_5$
& $x_0x_1x_3x_4x_5 + x_0x_3x_4x_5 + x_3x_4x_5$
& 59 & $x_2$
& $x_0x_1x_2 + x_1x_2 + x_2$ \\ \hline

\rowcolor{highlightcolor}5 & $x_1x_3x_4x_5$
& $x_0x_1x_2x_4x_5 + x_1x_3x_4x_5 + x_3x_4x_5$
& 58 & $x_0x_2$
& $x_0x_1x_2 + x_0x_2 + x_3$ \\ \hline

\rowcolor{highlightcolor}6 & $x_0x_3x_4x_5$
& $x_0x_1x_2x_4x_5 + x_0x_3x_4x_5 + x_1x_2x_4x_5$
& 57 & $x_1x_2$
& $x_0x_3 + x_1x_2 + x_3$ \\ \hline

7 & $x_3x_4x_5$
& $x_0x_2x_4x_5 + x_1x_2x_4x_5 + x_3x_4x_5$
& 56 & $x_0x_1x_2$
& $x_0x_1x_2 + x_0x_3 + x_1x_3$ \\ \hline

\rowcolor{highlightcolor}8 & $x_0x_1x_2x_4x_5$
& $x_0x_1x_2x_4x_5 + x_0x_2x_4x_5 + x_2x_4x_5$
& 55 & $x_3$
& $x_0x_1x_3 + x_1x_3 + x_3$ \\ \hline

\rowcolor{highlightcolor}9 & $x_1x_2x_4x_5$
& $x_0x_1x_4x_5 + x_1x_2x_4x_5 + x_2x_4x_5$
& 54 & $x_0x_3$
& $x_0x_1x_3 + x_0x_3 + x_2x_3$ \\ \hline

\rowcolor{highlightcolor}10 & $x_0x_2x_4x_5$
& $x_0x_1x_4x_5 + x_0x_2x_4x_5 + x_1x_4x_5$
& 53 & $x_1x_3$
& $x_0x_2x_3 + x_1x_3 + x_2x_3$ \\ \hline

11 & $x_2x_4x_5$
& $x_0x_4x_5 + x_1x_4x_5 + x_2x_4x_5$
& 52 & $x_0x_1x_3$
& $x_0x_1x_3 + x_0x_2x_3 + x_1x_2x_3$ \\ \hline

\rowcolor{highlightcolor}12 & $x_0x_1x_4x_5$
& $x_0x_1x_4x_5 + x_0x_4x_5 + x_4x_5$
& 51 & $x_2x_3$
& $x_0x_1x_2x_3 + x_1x_2x_3 + x_2x_3$ \\ \hline

13 & $x_1x_4x_5$
& $x_0x_1x_2x_3x_5 + x_1x_4x_5 + x_4x_5$
& 50 & $x_0x_2x_3$
& $x_0x_1x_2x_3 + x_0x_2x_3 + x_4$ \\ \hline

14 & $x_0x_4x_5$
& $x_0x_1x_2x_3x_5 + x_0x_4x_5 + x_1x_2x_3x_5$
& 49 & $x_1x_2x_3$
& $x_0x_4 + x_1x_2x_3 + x_4$ \\ \hline

15 & $x_4x_5$
& $x_0x_2x_3x_5 + x_1x_2x_3x_5 + x_4x_5$
& 48 & $x_0x_1x_2x_3$
& $x_0x_1x_2x_3 + x_0x_4 + x_1x_4$ \\ \hline

\rowcolor{highlightcolor}16 & $x_0x_1x_2x_3x_5$
& $x_0x_1x_2x_3x_5 + x_0x_2x_3x_5 + x_2x_3x_5$
& 47 & $x_4$
& $x_0x_1x_4 + x_1x_4 + x_4$ \\ \hline

\rowcolor{highlightcolor}17 & $x_1x_2x_3x_5$
& $x_0x_1x_3x_5 + x_1x_2x_3x_5 + x_2x_3x_5$
& 46 & $x_0x_4$
& $x_0x_1x_4 + x_0x_4 + x_2x_4$ \\ \hline

\rowcolor{highlightcolor}18 & $x_0x_2x_3x_5$
& $x_0x_1x_3x_5 + x_0x_2x_3x_5 + x_1x_3x_5$
& 45 & $x_1x_4$
& $x_0x_2x_4 + x_1x_4 + x_2x_4$ \\ \hline

19 & $x_2x_3x_5$
& $x_0x_3x_5 + x_1x_3x_5 + x_2x_3x_5$
& 44 & $x_0x_1x_4$
& $x_0x_1x_4 + x_0x_2x_4 + x_1x_2x_4$ \\ \hline

20 & $x_0x_1x_3x_5$
& $x_0x_1x_3x_5 + x_0x_3x_5 + x_3x_5$
& 43 & $x_2x_4$
& $x_0x_1x_2x_4 + x_1x_2x_4 + x_2x_4$ \\ \hline

21 & $x_1x_3x_5$
& $x_0x_1x_2x_5 + x_1x_3x_5 + x_3x_5$
& 42 & $x_0x_2x_4$
& $x_0x_1x_2x_4 + x_0x_2x_4 + x_3x_4$ \\ \hline

22 & $x_0x_3x_5$
& $x_0x_1x_2x_5 + x_0x_3x_5 + x_1x_2x_5$
& 41 & $x_1x_2x_4$
& $x_0x_3x_4 + x_1x_2x_4 + x_3x_4$ \\ \hline

23 & $x_3x_5$
& $x_0x_2x_5 + x_1x_2x_5 + x_3x_5$
& 40 & $x_0x_1x_2x_4$
& $x_0x_1x_2x_4 + x_0x_3x_4 + x_1x_3x_4$ \\ \hline

24 & $x_0x_1x_2x_5$
& $x_0x_1x_2x_5 + x_0x_2x_5 + x_2x_5$
& 39 & $x_3x_4$
& $x_0x_1x_3x_4 + x_1x_3x_4 + x_3x_4$ \\ \hline

25 & $x_1x_2x_5$
& $x_0x_1x_5 + x_1x_2x_5 + x_2x_5$
& 38 & $x_0x_3x_4$
& $x_0x_1x_3x_4 + x_0x_3x_4 + x_2x_3x_4$ \\ \hline

26 & $x_0x_2x_5$
& $x_0x_1x_5 + x_0x_2x_5 + x_1x_5$
& 37 & $x_1x_3x_4$
& $x_0x_2x_3x_4 + x_1x_3x_4 + x_2x_3x_4$ \\ \hline

27 & $x_2x_5$
& $x_0x_5 + x_1x_5 + x_2x_5$
& 36 & $x_0x_1x_3x_4$
& $x_0x_1x_3x_4 + x_0x_2x_3x_4 + x_1x_2x_3x_4$ \\ \hline

28 & $x_0x_1x_5$
& $x_0x_1x_5 + x_0x_5 + x_5$
& 35 & $x_2x_3x_4$
& $x_0x_1x_2x_3x_4 + x_1x_2x_3x_4 + x_2x_3x_4$ \\ \hline

29 & $x_1x_5$
& $x_0x_1x_2x_3x_4 + x_1x_5 + x_5$
& 34 & $x_0x_2x_3x_4$
& $x_0x_1x_2x_3x_4 + x_0x_2x_3x_4 + x_5$ \\ \hline

30 & $x_0x_5$
& $x_0x_1x_2x_3x_4 + x_0x_5 + x_1x_2x_3x_4$
& 33 & $x_1x_2x_3x_4$
& $x_0x_5 + x_1x_2x_3x_4 + x_5$ \\ \hline

31 & $x_5$
& $x_0x_2x_3x_4 + x_1x_2x_3x_4 + x_5$
& 32 & $x_0x_1x_2x_3x_4$
& $x_0x_1x_2x_3x_4 + x_0x_5 + x_1x_5$ \\ \hline

\rowcolor{highlightcolor} 32 & $x_0x_1x_2x_3x_4$ 
& $x_0x_1x_2x_3x_4 + x_0x_2x_3x_4 + x_2x_3x_4$ 
& 31 & $x_5$ & $x_5+x_1x_5+x_0x_1x_5$\\\hline

\rowcolor{highlightcolor} 33 & $x_1x_2x_3x_4$ 
& $x_0x_1x_3x_4 + x_1x_2x_3x_4 + x_2x_3x_4$ 
& 30 & $x_0x_5$ & $x_2x_5+x_0x_5+x_0x_1x_5$\\\hline

34 & $x_0x_2x_3x_4$ 
& $x_0x_1x_3x_4 + x_0x_2x_3x_4 + x_1x_3x_4$ 
& 29 & $x_1x_5$ & $x_2x_5+x_1x_5+x_0x_2x_5$\\
\bottomrule

\end{tabular}
}
    \caption{Algebraic form of the first 35 rows of $\bG$ (monomial code) and $\bP\bGN$ (upper polynomial code) where the convolution polynomial is defined by $[1,0,1,1].$ Red rows are those in $\F$ for the $[64,48]$ polar code. All remaining rows are in the information set $\forall\;i\geq 35, i\in\A.$ }
    \label{tab:PAC-polar-monomials}
\end{table*}

Polar codes and generalized polynomial polar codes do share common elements, i.e., their intersection is non-trivial. 
\begin{example}\label{ex:m6-pac-1}
    Let $m=6,N=64$ and $k=48.$ The frozen set in this case is \[\A^c=\{ 0, 1, 2, 3, 4, 5, 6, 8, 9, 10, 12, 16, 17, 18, 32, 33 \}.\] 

    Consider the PAC code defined by the polynomial $\bp=[1,0,1,1]$. The corresponding PAC code is a polynomial code with the following properties:
    \begin{itemize} 
        \item $\C_{\bP\bGN}(\A)$ contains a decreasing monomial subcode of dimension 30, defined by $\C([1,x_0x_2x_3x_4]_{\preceq}).$ %It can also be defined by the frozen set $\F=\{0,\dots,33\}.$
         \item $\C_{\bP\bGN}(\A)$ contains a monomial subcode of dimension 32, more exactly $\C([1,x_0x_2x_3x_4]_{\preceq})\cup \{x_2x_5,x_0x_1x_3x_5\}.$ The reason $x_2x_5$ belongs to the monomial set of the PAC code is because
         \begin{itemize}
             \item $x_2x_5$ was transformed by $\bp$ into $x_2x_5+x_1x_5+x_0x_5$
             \item $x_1x_5$ was transformed into $x_1x_5+x_5+x_0x_1x_2x_3x_4$
             \item $x_0x_5$ was transformed into 
             $x_0x_5+x_1x_2x_3x_4+x_0x_1x_2x_3x_4$
             \item $x_5$ was transformed into 
             $x_5+x_1x_2x_3x_4+x_0x_2x_3x_4$
         \end{itemize}  
         Overall we get that $x_2x_5+x_1x_5+x_0x_5+x_5$ was transformed by $\bp$ into $x_2x_5+x_0x_2x_3x_4$, and since $x_0x_2x_3x_4$ is already a monomial then $x_2x_5$ has to be a monomial of the PAC code. The same can be deduced for $x_0x_1x_3x_5.$
          \item The intersection code $\C_{\bG}(\A)\cap\C_{\bP\bGN}(\A)$ is a code with dimension 42. 
    \end{itemize}
\end{example}

\subsection{Upper polynomial polar codes}
One can notice that any upper polynomial polar code contains a decreasing monomial code, with dimension $N-1-\max(\A^c)$, where $\max(\A^c)$ is the index of the last frozen monomial. In Example \ref{ex:m6-pac-1} the last frozen index was $33$ and $N-1-33=63-33=30.$ However, some upper polynomial polar code might contain larger decreasing monomial codes. 

\begin{figure}[!ht]
    \centering 
    \resizebox{0.45\textwidth}{!}{
    \begin{tikzpicture}
            \begin{axis}[width=\textwidth,
                    ylabel={$\dim(\C_{\bGN}(\A)\cap \C_{\bP\bGN}(\A))$},
                    xlabel={$k$},
                    ytick={10,15,20,25,30,35,40},
                    ymin=7.9,
                    ymax=41.1,
                    xtick={15,20,25,30,35,40,45,50},
                    xmin=13.9,
                    xmax=50.1,
                    legend style={at={(0.05,0.95)},anchor=north west},
                    grid=both
                    ]
            \addplot+[solid,mark=none,red] table [x, y, col sep=comma]{FiguresISIT2026/Generalized-Avr-polar.txt};\addlegendentry{Average Generalized polynomial}\label{average-polar}
            \addplot+[solid,mark=o,blue] table [x, y, col sep=comma]{FiguresISIT2026/f11-polar.txt};\addlegendentry{$\bp=[1,1]$}\label{p11-monomial}
             \addplot+[solid,mark=diamond,blue] table [x, y, col sep=comma]{FiguresISIT2026/f101-polar.txt};\addlegendentry{$\bp=[1,0,1]$}\label{p101-monomial}
              \addplot+[solid,mark=star,blue] table [x, y, col sep=comma]{FiguresISIT2026/f1011-polar.txt};\addlegendentry{$\bp=[1,0,1,1]$}\label{p1011-monomial}
               \addplot+[solid,mark=oplus,blue] table [x, y, col sep=comma]{FiguresISIT2026/f1111-polar.txt};\addlegendentry{$\bp=[1,1,1,1]$}\label{p1111-monomial}
                \addplot+[solid,mark=square,blue] table [x, y, col sep=comma]{FiguresISIT2026/f10101-polar.txt};\addlegendentry{$\bp=[1,0,1,0,1]$}\label{p10101-monomial}
                 \addplot+[solid,mark=triangle,blue] table [x, y, col sep=comma]{FiguresISIT2026/f1011011-polar.txt};\addlegendentry{$\bp=[1,0,1,1,0,1,1]$}\label{p1011011-monomial}
                  \addplot+[solid,mark=10-pointed star,blue] table [x, y, col sep=comma]{FiguresISIT2026/f101010101-polar.txt};\addlegendentry{$\bp=[1,0,1,0,1,0,1,0,1]$}\label{p101010101-monomial}
                  \end{axis}
        \end{tikzpicture}
        }
     \caption{Dimension of the intersection code between the PAC code and polar codes for $N=64$ and $14\leq k \leq 50.$}\label{fig:PAC-polar-intersect}
\end{figure}
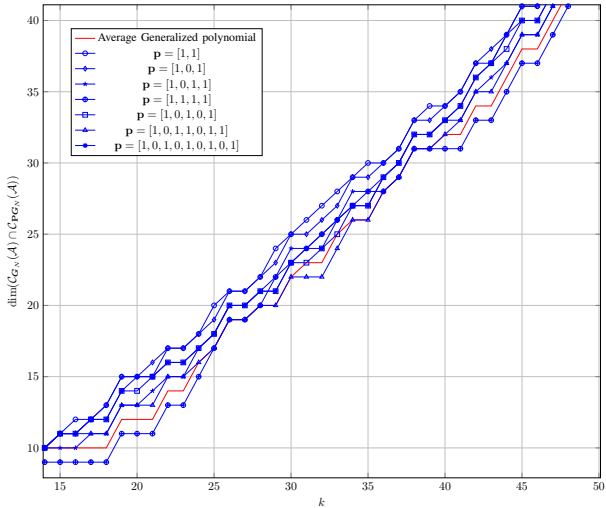

% \begin{definition}
%     Let $m,k$ be integers and $\A\subset [0, 2^m-1]$ be the information of a pure polar code of length $N$ and dimension $K.$  Define $\A_{c}=[\max(\A^c)+1, N-1].$  
% \end{definition}

% Notice that we have $\A_c\subseteq \A^c.$ The following result holds.

\begin{proposition}
    Let $\C_{\bP\bGN}(\A)$ be a upper polynomial polar code. Then, $\C_{\bG}([\max(\A^c)+1, N-1])$ is a decreasing monomial code with dimension $N-1-\max(\A^c)$ satisfying $\C_{\bG}([\max(\A^c)+1, N-1])\subset \C_{\bP\bGN}(\A).$  
\end{proposition}

\begin{proof}
First let us denote the set of monomials corresponding to $\A$ by $\I$ and the set of monomials corresponding to $[\max(\A^c)+1, N-1]$ by $\I_{\F}.$
The dimension of $\C_{\bG}(\A_{c})$ equals $N-1-\max(\A^c)$ since all monomials in the set $\I_{\F}=\{g\mid \ev(g)=\bGN[i],i \in \A_{c}\}$ are linearly independent. Let us prove that $\C_{\bG}(\A_{c})$ is a decreasing monomial code. By definition, the monomial basis of $\C_{\bG}(\A_{c})$ consists of all monomials corresponding to consecutive integers in the set $\A_{c}.$  Now, suppose there is a monomial $f\preceq g$ with $f\not\in\I_{\F}$ and $g\in \I_{\F}.$ Notice that we can not have $f\preceq_w g$, since this would imply $\ind(f)\subset \ind(g)\Rightarrow 2^m-1-i_f<2^m-1-i_g\Rightarrow i_f>i_g$ which contradicts $f\not\in\I_{\F}$ or equivalently $i_f\not \in [\max(\A^c)+1, N-1].$ If $f\preceq_{\mathbf{sh}} g$, by definition we have $2^m-1-i_f<2^m-1-i_g\Rightarrow i_f>i_g$ which again contradicts $f\not\in\I_{\F}.$        
\end{proof}

We have conducted simulations on the polar code of length $64$ and generalized polynomial polar codes. We have computed the intersection of the plain polar codes and the PAC codes reported in Figure \ref{fig:PAC-polar-intersect} as well as set of $1000$ random upper polynomial polar codes. We have computed the intersection code $\C_{\bGN}(\A)\cap\C_{\bP\bGN}(\A)$, which is a polynomial code, it contains a monomial basis plus some polynomials/codewords (which are not evaluation of monomials) from the polar code. For the $1000$ random polynomial polar codes we have computed the average dimension and reported by the red straight line. Notice that most of the PAC codes have an intersection above the average of generalized polynomial polar codes. In other words, an upper polynomial polar code is more likely to break the polar structure, except for some specific convolution polynomials $\bp.$   

PAC codes are sub-codes of some higher dimension monomial codes. The dimension of these monomial codes depending on the convolutional polynomials as well as the code rates. However, there is an upper bound on their dimension. 

\begin{lemma}\label{lem:max-monom}
    Let $\C_{\bP\bGN}(\A)$ be a upper polynomial polar code. Then, $\C_{\bP\bGN}(\A)\subset \C_{\bG}([\min(\A),N-1]).$ 
\end{lemma}
The proof is trivial and comes form the definition of a polynomial polar code. Notice that the dimension of this monomial code $\C_{\bG}([\min(\A),N-1])$ equals $N-\min(\A).$ 

\begin{figure}[h]
\centering
\resizebox{0.48\textwidth}{!}{
          \begin{tikzpicture}
            \begin{axis}[width=\textwidth,
                    ylabel={},%{$\dim(\C(\I)),\I=\bigcup\limits_{\ev(f)\in\C_{\bP\bGN(\A)}}\terms(f)$},
                    xlabel={$k$},
                    ytick={5,10,15,20,25,30,35,40,45,50,55,60},
                    ymin=4.9,
                    ymax=61,
                    xtick={14,20,25,30,35,40,45,50},
                    xmin=13.9,
                    xmax=50.1,
                    legend style={at={(0,1)},
                    anchor=north west},
                    grid=both
                    ]
                    \addplot+[solid,mark=10-pointed star,red] table [x, y, col sep=comma]{FiguresISIT2026/UBound-decreasing.txt};\addlegendentry{$N-\min(\A)$}\label{ubound-decreasing}
            \addplot+[solid,mark=none,red] table [x, y, col sep=comma]{FiguresISIT2026/Generalized-Avr-monomial-up.txt};\addlegendentry{Avr. Gen. poly.}\label{average-monomial}
            \addplot+[solid,mark=o,blue] table [x, y, col sep=comma]{FiguresISIT2026/f11-monomial-up.txt};\addlegendentry{$\bp=[1,1]$}\label{p11-monomial}
             %\addplot+[mark=diamond,blue] table [x, y, col sep=comma]{FiguresISIT2026/f101-monomial.txt};\addlegendentry{$\bp=[1,0,1]$}\label{p101-monomial}
              \addplot+[solid,mark=star,blue] table [x, y, col sep=comma]{FiguresISIT2026/f1011-monomial-up.txt};\addlegendentry{$\bp=[1,0,1,1]$}\label{p1011-monomial}
              % \addplot+[mark=oplus,blue] table [x, y, col sep=comma]{FiguresISIT2026/f1111-monomial.txt};\addlegendentry{$\bp=[1,1,1,1]$}\label{p1111-monomial}
               % \addplot+[mark=square,blue] table [x, y, col sep=comma]{FiguresISIT2026/f10101-monomial.txt};\addlegendentry{$\bp=[1,0,1,0,1]$}\label{p10101-monomial}
                 \addplot+[solid,mark=triangle,blue] table [x, y, col sep=comma]{FiguresISIT2026/f1011011-monomial-up.txt};\addlegendentry{$\bp=[1,0,1,1,0,1,1]$}\label{p1011011-monomial}
                  \addplot+[solid,mark=10-pointed star,blue] table [x, y, col sep=comma]{FiguresISIT2026/f101010101-monomial-up.txt};\addlegendentry{$\bp=[1,0,1,0,1,0,1,0,1]$}\label{p101010101-monomial}
                  % lower part -------------------------
                  \addplot+[dashed,mark=10-pointed star,red] table [x, y, col sep=comma]{FiguresISIT2026/LBound-decreasing.txt};\addlegendentry{$N-1-\max(\A_c)$}\label{lbound-decreasing}
                    \addplot+[dashed,mark=none,red] table [x, y, col sep=comma]{FiguresISIT2026/Generalized-Avr-monomial.txt};\addlegendentry{Avr. Gen. poly.}\label{average-monomial}
            \addplot+[dashed,mark=o,blue] table [x, y, col sep=comma]{FiguresISIT2026/f11-monomial.txt};\addlegendentry{$\bp=[1,1]$}\label{p11-monomial}
             %\addplot+[mark=diamond,blue] table [x, y, col sep=comma]{FiguresISIT2026/f101-monomial.txt};\addlegendentry{$\bp=[1,0,1]$}\label{p101-monomial}
              \addplot+[dashed,mark=star,blue] table [x, y, col sep=comma]{FiguresISIT2026/f1011-monomial.txt};\addlegendentry{$\bp=[1,0,1,1]$}\label{p1011-monomial}
              % \addplot+[mark=oplus,blue] table [x, y, col sep=comma]{FiguresISIT2026/f1111-monomial.txt};\addlegendentry{$\bp=[1,1,1,1]$}\label{p1111-monomial}
               % \addplot+[mark=square,blue] table [x, y, col sep=comma]{FiguresISIT2026/f10101-monomial.txt};\addlegendentry{$\bp=[1,0,1,0,1]$}\label{p10101-monomial}
                 \addplot+[dashed,mark=triangle,blue] table [x, y, col sep=comma]{FiguresISIT2026/f1011011-monomial.txt};\addlegendentry{$\bp=[1,0,1,1,0,1,1]$}\label{p1011011-monomial}
                  \addplot+[dashed,mark=10-pointed star,blue] table [x, y, col sep=comma]{FiguresISIT2026/f101010101-monomial.txt};\addlegendentry{$\bp=[1,0,1,0,1,0,1,0,1]$}\label{p101010101-monomial}
                   \addplot+[mark=none,green] table [x, y, col sep=comma]{FiguresISIT2026/diagonal.txt};\addlegendentry{$k$}\label{k}
                  \end{axis}
        \end{tikzpicture}
        }
        \caption{Dimension of the smallest monomial code that contains a PAC code (solid lines, above the green line) and the dimension of the largest monomial sub-code of the PAC code (dashed lines, below the green line), for $n=64$ and $14\leq k\leq 50.$}
    \label{fig:PAC-monomial}
     \end{figure}
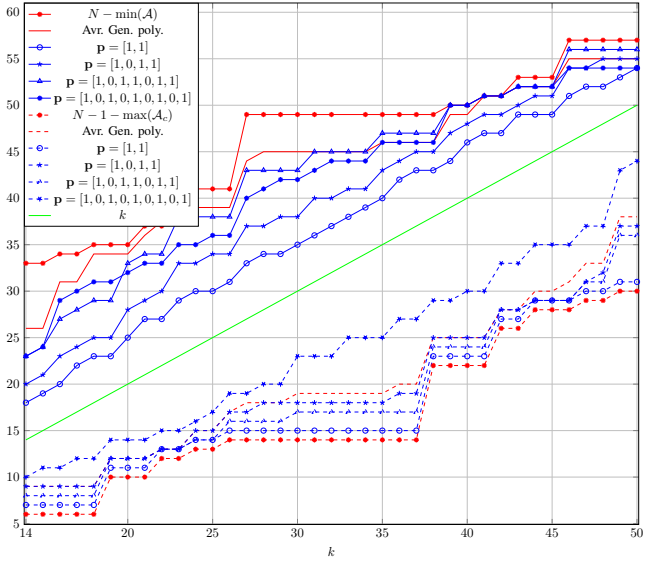

For the second set of simulations we have computes the number of monomials contained in the upper polynomial polar code, results reported in Figure \ref{fig:PAC-monomial} - dashed lines. We have observed that large convolutional polynomials with regular structure tend to maximize the dimension of the monomial sub-code of a PAC code. Notice that the largest values for the monomial sub-code dimension is for $\bp=[1,0,1,0\dots,1,0,1,0,1].$ This motivates to search for particular convolutional polynomials that might break the plain polar structure as much as possible, e.g., $\bp=[1,1]$ or $\bp=[1,0,1,1,0,1,1].$ On the same figure, represented by solid lines, we also illustrate the dimension of the smallest monomial code that contains the PAC code. Formally we compute $\dim(\C(\I))$ where $\I=\bigcup_{\ev(f)\in\C_{\bP\bGN(\A)}}\terms(f).$ Upper polynomial polar codes usually tend to approach the maximum limit reported in Lemma \ref{lem:max-monom}. Also, the higher the Hamming weight of $\bp$ the closer we get to this limit, which is expected since we are involving more monomials in the polynomial basis of the polynomial polar code.   

\section{Duality}
\subsection{Properties of duals of monomial codes}

From \cite{bardet}, we use the multiplicative complement of a monomial $g\in \Mon$ defined by $\widecheck{g}\triangleq \frac{x_0\dots x_{m-1}}{g}.$ Remark that it is well-defined since all monomials divide $x_0\dots x_{m-1}.$ Note that while $g=\mathbf{x}^{\bin(2^m-1-i)}$ for a given $i$, its complement is $(\widecheck{g})=\mathbf{x}^{\bin(i)}.$ We can extend the definition to a monomial set $\widecheck{\I}=\{\widecheck{g}\mid g\in \I\}$ and to polynomials as well, $\widecheck{Q}=\sum\limits_{f\in \terms(Q)}\widecheck{f}$ for any $Q\in\Rm$. We have
\[
\begin{array}{cccc}
    g&\rightarrow \widecheck{g}=\dfrac{x_0\dots x_{m-1}}{g}& \quad \I& \rightarrow \widecheck{\I}=\{\widecheck{g}\mid g \in \I\}\\
    i&\rightarrow \widecheck{i}=2^m-1-i& \quad \A& \rightarrow \widecheck{\A}=\{\widecheck{i}\mid i \in \A\}
\end{array}
\]  
We know from \cite{bardet} that the dual of a decreasing monomial code is still a decreasing monomial code: 
$\C(\I)^{\bot}=\C(\Mon\setminus \widecheck{\I})$. We can rewrite the equality $\C(\I)^{\bot}=\C(\widecheck{\F}).$ Indeed, $\Mon\setminus\widecheck{I}=\widecheck{\Mon}\setminus\widecheck{\I}=\widecheck{\left(\Mon\setminus\I\right)}=\widecheck{\F}.$ Or, equivalently 
\begin{equation}\label{eq:dual-monomial-dragoi}
    \C_{\bGN}(\A)^{\bot}=\C_{\bGN}(\widecheck{\A^c}).
\end{equation}

However, the dual of a monomial code (where $\I$ is simply a monomial set) is in general a polynomial code \cite{dragoi17thesis}. \[\C(\I)^{\bot}=\C(\psi(\widecheck{\F})),\]

where $\psi(\I)=\{\psi(f)\mid f\in \I\}$ for any monomial set $\I$, and $\psi(x_{i_1}\dots x_{i_s})=(1+x_{i_1})\cdot \dots \cdot (1+x_{i_s}).$ The action of $\psi$ can also be described in terms of permutations. Indeed, if we denote by $\pi$ (or in matrix terms $\bQ_{\pi}$) the permutation that swaps $i$ with $2^m-1-i$ for all indices $i\in[0,2^m-1]$, then, we have \cite{dragoi17thesis}
\begin{equation}\label{eq:psi-permutation-dragoi}
    \ev(\psi(g))^{\pi}=\ev(g)
\end{equation}

The following properties hold.

\begin{proposition}\label{pr:dual-monomial-A-form}
Let $\A\subset [0,N-1]$ and $\C_{\bGN}(\A)$ be a monomial code. We have $\bQ_{\pi}=\bQ_{\pi}^{-1}=\bQ_{\pi}^{t}$ and $\C_{\bGN}(\A)^{\bot}=\C_{\bGN\bQ_{\pi}}(\widecheck{\A^c}).$
    % \begin{itemize}
    %     \item $\bQ_{\pi}=\bQ_{\pi}^{-1}=\bQ_{\pi}^{T}$;
    %     \item $\C_{\bGN}(\A)^{\bot}=\C_{\bGN\bQ_{\pi}}(\widecheck{\A^c})$
    % \end{itemize}
\end{proposition}

\begin{proof}
    The first property is straightforward from the definition of $\bQ_{\pi}.$ For the second we have
    \begin{align*}
        \C_{\bGN}(\A)^{\bot}&\overset{\eqref{eq:dual-monomial-dragoi}}{=}\C(\I)^{\bot}\overset{\eqref{eq:psi-permutation-dragoi}}{=}\C(\psi(\widecheck{\F}))=\C((\widecheck{\F}))^{\pi}=\C_{\bGN\bQ_{\pi}}(\widecheck{\A^c}).
    \end{align*}
     We have used the mapping $\A^c\leftrightarrow \F;\A\leftrightarrow\I.$ 
\end{proof}

Algebraic properties of $\bG_N$ and $\bQ_{\pi}$ will be at hand. 

\begin{lemma}\label{lem:prop-perm-psi-G}
   Let $\bQ_{\pi}$ be the matrix permutation of $\psi.$ Then %$\bGN=\bGN^{-1}$ and $(\bGN\bQ_{\pi})^2=\bQ_{\pi}\bGN$ %, we have
    \begin{itemize}
        \item $\bGN=\bGN^{-1}$
        \item $\bQ_{\pi}\bGN=(\bG\bQ_{\pi})^{-1}$;
        \item $(\bGN\bQ_{\pi})^2=\bQ_{\pi}\bGN$
        \item $\bGN\bQ_{\pi}\bGN=\bQ_{\pi}\bGN\bQ_{\pi}.$
    \end{itemize}
\end{lemma}

We can deduce from the previous lemma that $\bGN\bQ_{\pi}\bGN=\bQ_{\pi}\bGN\bQ_{\pi}.$
\begin{proof}
    The first identity comes from the fact that \[\begin{pmatrix}
        \bA & \bm{0}\\
        \bA & \bA
    \end{pmatrix}\cdot \begin{pmatrix}
        \bA & \bm{0}\\
        \bA & \bA
    \end{pmatrix}=\begin{pmatrix}
        \bA^2 & \bm{0}\\
        \bm{0} & \bA^2
    \end{pmatrix}\] for any matrix $\bA$ defined over a field of characteristic 2. Apply this property to $\bGN$ and use the fact that $\bG_2^2=\bId$ to obtain the result.  
    The second identity follows immediately from the definition of $\bQ_{\pi}$ and $\bGN^{-1}=\bGN.$ 
    
    Regarding the third identity we have that \[\begin{pmatrix}
        \bm{0}& \bA \\
        \bA & \bA
    \end{pmatrix}\cdot \begin{pmatrix}
        \bm{0} & \bA \\
        \bA & \bA
    \end{pmatrix}\cdot \begin{pmatrix}
        \bm{0} & \bA \\
        \bA & \bA
    \end{pmatrix}=\begin{pmatrix}
        \bA^2 &\bm{0} \\
        \bm{0} &\bA^2 
    \end{pmatrix}\] for any matrix $\bA$ defined over a field of characteristic 2. Since $\bGN\bQ_{\pi}$ is of the form  $\begin{pmatrix}
        \bm{0}& \bA \\
        \bA & \bA
    \end{pmatrix}$ with $\bA^2=\bId$ we obtain the wanted result.

    The last identity follows directly from the third identity and the first one.
\end{proof}

\begin{remark}
Notice that row $i$ of $\bQ_{\pi}\bGN\bQ_{\pi}$ is equal to $\ev(\psi(\mathbf{x}^{\bin(i)})).$ Also, each column of $\bGN$ is the evaluation of $\psi(\mathbf{x}^{\bin(i)}).$
\end{remark}

A direct consequence of Lemma \ref{lem:prop-perm-psi-G} and Proposition \ref{pr:dual-monomial-A-form} is the following.

\begin{corollary}
    Let $\I\subset \Mon$ be a monomial set. We have
    \begin{equation}
        \C(\I)^{\bot}=\C_{\bGN\bQ_{\pi}\bGN}(\A^c).
    \end{equation}
\end{corollary}

\begin{proof}
    First notice that $\C_{\bGN}(\widecheck{\A})=\C_{\bQ_{\pi}\bGN}(\A)$, due to the fact $i$ becomes $2^m-1-i$ by $\bQ_{\pi}$ and $i$ becomes $2^m-i-1$ by $\widecheck{i}.$ Hence, taking information set $\widecheck{\A}$ is equivalent to taking information set $\A$ on the permuted matrix $\bQ_{\pi}\bGN.$

    Therefore, we have 
    \begin{equation*}
        \C_{\bGN\bQ_{\pi}}(\widecheck{\A^c})=\C_{\bQ_{\pi}\bGN\bQ_{\pi}}(\A^c).
    \end{equation*}

    Using the properties of $\bGN$ and $\bQ_{\pi}$ from \ref{lem:prop-perm-psi-G} we have 
    \begin{equation*}
        \C_{\bQ_{\pi}\bGN\bQ_{\pi}}(\A^c)=\C_{\bGN\bQ_{\pi}\bGN}(\A^c),
    \end{equation*}
    which ends our proof.
\end{proof}

Our results reads as follows. Given a monomial code, specified by an information set $\A$, compute the frozen set $\A^c$ which will be the information of the dual code. The dual code is a lower polynomial code since $\bQ_{\pi}\bGN$ is a monomial basis and $\bGN$ is a lower triangular matrix. %no longer monomial, hence it is generated by a polynomial basis, more exactly, $\bGN\bQ_{\pi}\bGN.$ 

% \begin{remark}
%     Since $\bQ_{\pi}\bGN$ is also a monomial matrix (monomials are indexed differently from $\bGN$) and since $\bGN$ is an invertible lower triangular matrix, the dual of a monomial code with information set $\A$ is a lower polynomial code with information set $\A^c$. 
% \end{remark}

\subsection{Dual of polynomial polar codes}
When dealing with polynomial polar codes, one can not simply select the complement of monomials in the frozen set $\F$ since these were previously considered in combination with monomials in $\I$, when the code was generated. 
\begin{example}
    Take the $N=64,k=48$ PAC code with $\bp=[1,0,1,1].$ Observe that the frozen monomial $g=x_0x_1x_4x_5$ becomes $P_g=x_0x_1x_4x_5+x_0x_4x_5+x_4x_5.$ Hence, when taking the complement we obtain $\widecheck{g}=x_2x_3$ and $\widecheck{P_g}=x_2x_3+x_1x_2x_3+x_0x_1x_2x_3.$ Even if we apply $\psi(\widecheck{P_g})$ there is no guaranty that $\ev(\psi(\widecheck{P_g}))$ is orthogonal of all vectors in the PAC code. For example, $\psi(\widecheck{P_g})$ is not orthogonal on $x_4x_5\in\I.$ 
    The main reason for that is because when we multiply these two polynomials we obtain a new polynomial which has among its terms the monomial $x_0\cdots x_{m-1}.$ Indeed, we have 
\begin{align*}
    \psi(\widecheck{P_g})\cdot x_4x_5&=\psi(x_2x_3+x_1x_2x_3+x_0x_1x_2x_3)\cdot x_4x_5\\
    &=(x_2+1)(x_3+1)\cdot x_4x_5\\
    &+(x_1+1)(x_2+1)(x_3+1)\cdot x_4x_5\\
    &+(x_0+1)(x_1+1)(x_2+1)(x_3+1)\cdot x_4x_5.
\end{align*}  

After expanding the terms we do obtain that $x_0x_1x_2x_3x_4x_5\in\terms( \psi(\widecheck{P_g})\cdot x_4x_5)$, which implies that $\ev (\psi(\widecheck{P_g})),\ev(x_4x_5)$ are not orthogonal.
\end{example}

%The previous example illustrated a crucial property on orthogonality.
A crucial property on orthogonality can be expressed in terms of product of polynomials. 
\begin{lemma}\label{lem:duality-cond-on-max-mon}
    Let $P_1,P_2\in\Rm.$ Then $\ev(P_1)\cdot \ev(P_2)=0$ iff $x_0\dots x_{m-1}\not \in \terms(P_1\cdot P_2).$
\end{lemma}

\begin{proof}
    By definition of the scalar product between two vectors $\ev(P_1)$ and $\ev(P_2)$ we have 
    \begin{align*}
        \ev(P_1)\cdot \ev(P_2)&=\sum_{i=0}^N \left(\ev(P_1)\odot \ev(P_2)\right)[i]=\sum_{i=0}^N \ev(P_1\cdot P_2)[i].
    \end{align*}

 Hence, as long as $\ev(P_1\cdot P_2)$ is of even weight the scalar product equals zero, due to the fact that summation is over $\ft.$ Since the code $\R(m-1,m)$ is the code of all possible vectors of even Hamming weight, and $x_0\cdots x_{m-1}$ is the only monomial in $\R(m,m)\setminus\R(m-1,m)$ we deduce the wanted result.  
\end{proof}

We can explain how to compute the dual of upper polynomial polar code and reveal its structure. Our idea builds on the use of the multiplicative complement and $\psi.$ We will thus search monomials $f$ that are indexed by the frozen set $\A^c$ and determine conditions for orthogonality with the polynomials in the basis of $\C_{\bP\bGN}(\A).$

The first property that we exhibit states that when searching for duality conditions for upper polynomial polar codes we only need to look backwards in terms of indices. 

\begin{lemma}\label{lem:forward-not}
    Let $\C_{\bP\bGN}(\A)$ be an upper polynomial polar code. Let $\ev(f)=\bGN[i]$ with $i\in[0,2^m-1].$ We have \[\forall j>i \quad \ev(\widecheck{f})\cdot \bP\bGN[j]=0.\]
    
\end{lemma}

\begin{proof}
    Let $P=\bP\bGN[j]$. By definition $P=\bGN[j]+\sum\limits_{l\in J\setminus\{j\}}\bGN[l]$ where $J$ is the set corresponding to row $j$ of $\bP$, i.e., $J=\supp(\bP[j]).$ We thus have 
    \[\ev(P)\cdot \ev(\widecheck{f})=\sum\limits_{k=0}^N\ev((\bGN[j]+\sum\limits_{l\in J\setminus\{j\}}\bGN[l])\cdot \psi(\widecheck{f})).\]
    In order to have a non zero scalar product, by Lemma \ref{lem:duality-cond-on-max-mon}, the monomial $x_0\cdots x_{m-1}$ should belong to the terms of $(\bGN[j]+\sum\limits_{l\in J}\bGN[l])\cdot \psi(\widecheck{f}).$ This implies that there are some indices $l\in J$ for which $\widecheck{f}\cdot \bGN[l]=x_0\cdots x_{m-1}$ which is equivalent to  $\widecheck{\bGN[i]}\cdot \bGN[l]=x_0\cdots x_{m-1}.$ For this equation to be satisfied we must either have $\bGN[i]=\bGN[l]$ or $\bGN[i]|\bGN[l]$ which contradicts $j>i.$ 
\end{proof}

Hence, the only way of combining polynomials from a set $\A$ with multiplicative complements of monomials from frozen set $\A^c$ such that they do not define orthogonal vectors is by combining monomials backwards. That is why we will define a matrix which acts exactly in the opposite direction of $\bP.$ To do so, the following property will be at use.
\begin{lemma}\label{lema:backward-parity-condition}   
Let $\C_{\bP\bGN}(\A)$ be a upper polynomial polar code. Let $\ev(f)=\bGN[i]$ for an integer $i\in[0,2^m-1].$ Then for any $j<i$ we have \[g\in\terms(\bP\bGN[j]) \text{ and }x_0\cdots x_{m-1}= g\cdot \widecheck{f} \Rightarrow \widecheck{g}|\widecheck{f}.\]
\end{lemma}

\begin{proof}
    For $j<i$ the terms in $\bP\bGN[j]$ are monomials $g$ with $\ev(g)=\bGN[l]$ for some values $l<i.$ Thus,
    \[x_0\cdots x_{m-1}=g\cdot \widecheck{f} \Rightarrow \widecheck{f} \text{ is a multiple of } \dfrac{x_0\cdot x_{m-1}}{g}=\widecheck{g}.\]
\end{proof}

We can now define a matrix $\bM$ that counts the number of monomials/terms in each polynomial $P$, $\ev(P)=\bP\bGN[j]$, which fall under the previous condition. More exactly for each monomial $f$, $\ev(f)=\bGN[i]$ we will count how may $g\in\terms(P)$ satisfy the condition $\widecheck{g}|\widecheck{f}.$ 

\begin{equation}\label{eq:matrix-freq}
        M_{i,j}=\left|\{g \in \terms(P), \mid \widecheck{g}| \widecheck{f} \}\right| \pmod{2}
    \end{equation}
where $\ev(P)=\bP\bGN[j], \ev(f)=\bGN[i].$

Given a row index $i$, recall that our main issue comes from row indices $j<i.$ Hence, if we consider a monomial $\widecheck{f}$, where $\ev(f)=\bGN[i]$ the vector in the code $\C_{\bP\bGN}(\A)$ raising problems with respect to orthogonality are exactly those coming from indices $j<i.$ If there are an even number of monomials $g\in\terms(P)$, $P$ being the polynomial corresponding to a codeword $\bP\bGN[j]$, they cancel out each other due to $2=0$ over $\ft.$ The problem is when there are an odd number of such monomials. In this case we need to combine our initial monomial $f$ with another one that manages to cancel out the $x_0\cdots x_{m-1}$ resulting from $g\cdot \widecheck{f}.$ Let's see all of these on an example.  

\begin{example}
    Take the $(64,48)$ polar code and the PAC code resulting using convolutional polynomial $[1,0,1,1].$ To simplify notations the row $\widecheck{P}$ will be indexed as $\widecheck{P_i}.$  
    \begin{itemize}
        \item Row $0:$ $\widecheck{f}=1\rightarrow [1,0,\dots ]$, since row 0 of $\terms(\widecheck{P_0})=\{1,x_1,x_0x_1\}.$
        \item Row $1:$ $\widecheck{f}=x_0\rightarrow [1,1, 0,\dots ].$ Here we have one term in row 0 of $P_0$ which divides $x_0$, namely $1|x_0$ while in $\terms(P_1)=\{x_0,x_0x_1,x_2\}$ we have one term $x_0|x_0.$ 
        \item Row $2:$ $\widecheck{f}=x_1\rightarrow [0 ,0, 1, 0, \dots ].$ Here we have two terms in row 0 of $P_0$ which divides $x_1$, namely $1,x_1$; zero terms in $\terms(P_1)=\{x_0,x_0x_1,x_2\}$; and one term in $\terms(P_2)=\{x_1,x_0x_2,x_2\}.$ 
        \item Row $3:$ $\widecheck{f}=x_0x_1\rightarrow [1 ,0, 1, 1, 0, \dots ].$ Here we have three terms in row 0 of $P_0$ which divides $x_1$, namely $1,x_1,x_0x_1$; two terms in $P_1$, namely $x_0,x_0x_1$; one term in $P_2$, i.e., $x_1$; and one term in $P_3.$
    \end{itemize}
    Continuing up to the first 10 rows we obtain \[\begin{pmatrix}
        1 &0 &0 &0 &0 &0 &0 &0 &0 &0\\
1 &1 &0 &0 &0 &0 &0 &0 &0 &0\\
0 &0 &1 &0 &0 &0 &0 &0 &0 &0\\
1 &0 &1 &1 &0 &0 &0 &0 &0 &0\\
1 &1 &1 &0 &1 &0 &0 &0 &0 &0\\
1 &0 &0 &1 &1 &1 &0 &0 &0 &0\\
0 &1 &0 &1 &0 &0 &1 &0 &0 &0\\
1 &1 &1 &1 &1 &0 &1 &1 &0 &0\\
1 &0 &0 &0 &0 &1 &1 &0 &1 &0\\
1 &1 &0 &0 &0 &1 &0 &1 &1 &1\\
    \end{pmatrix}\]
\end{example}

Before we reveal the dual of an upper polynomial polar code we will detail some properties regarding the matrix $\bM.$
\paragraph{Properties of $\bM$}
First, notice that each row of our matrix $\bM$ is indexed by $\bx^{\bin(i)}$ with $i=[0,2^{m}-1]$, while the columns are indexed by $\widecheck{P}$ where $\ev(P)=\bP\bGN[2^m-i].$ 

Another interesting property is when $\bP=\bId$.

\begin{lemma}\label{lem:M-trivial}
    Let $\bP=\bId.$ Then we have 
    \[\bM=\bGN.\]
\end{lemma}

 This property reveals another meaning of the generator matrix of plain polar codes. Its rows express the divisibility between a fixed monomial $f$ and the set of all monomials $\Mon$ taken in a particular order. 

\begin{example}
   Let $m=3$ and $\bP=\bId$ (see Table \ref{tab:contruct-M-id}). 
    \begin{table}[!h]
        \centering
        \resizebox{0.45\textwidth}{!}{
        \begin{tabular}{c||c|c|c|c| c|c|c|c}
        \toprule
       $f$ &\multicolumn{8}{c}{$g$}\\
        &$x_0x_1x_2$ & $x_1x_2$ & $x_0x_2$ & $x_2$ & $x_0x_1$ & $x_1$ & $x_0$ & $1$\\
        \midrule
        $x_0x_1x_2$ & 1 & - & - & - & - &- &- & -\\
        $x_1x_2$ & $x_0$ & 1 & - & - & - &- &- & -\\
        $x_0x_2$ & $x_1$ & - & 1 & - & - &- &- & -\\
            $x_2$ & $x_0x_1$ & $x_1$ & $x_0$ & 1 & - &- &- & -\\ 
            $x_0x_1$ & $x_2$ & - & - & - & 1 & - & - & - \\
             $x_1$ & $x_0x_2$ & $x_2$ & - & - & $x_0$ & 1 & - & - \\
             $x_0$ & $x_1x_2$ & - & $x_2$& - & $x_1$ & - & 1& - \\
             $1$ &$x_0x_1x_2$ & $x_1x_2$ & $x_0x_2$ & $x_2$ &  $x_0x_1$ & $x_1$ & $x_0$ & 1\\
             \bottomrule
        \end{tabular}
        }
        \caption{ Each elements represents the column monomial divided by the row monomial $g/f$ when possible; while when divisibility is not possible the symbol is $-$ used. }
        \label{tab:contruct-M-id}
    \end{table}

    If we replace now each time we obtain a monomial with $1$ and each time we obtain $-$ with $0$ we retrieve the matrix $\bM$, which in this case is equal to $\bG_{2^3}.$
    \end{example}
More significant, the matrix $\bM$ is invertible.

\begin{lemma}
   $\bM$ is a lower triangular square matrix, with $M_{i,i}=1$ for all $i\in[0,2^{m}-1].$ Thus, $\bM$ is invertible over $\ft.$ 
\end{lemma}     

\begin{proof}
    It is straightforward using Lemma \ref{lem:forward-not} and \ref{lema:backward-parity-condition} that $\bM$ is lower triangular. The fact that $\bM$ has only ones on the diagonal comes from Lemma \ref{lem:forward-not} and the definition of $\bP.$  
\end{proof}

Using all these we can infer a bit more on this matrix. 
\begin{proposition}\label{pr:M-GP}
    \begin{equation}
        \bM=\bGN\bP^t.
    \end{equation}
\end{proposition}

\begin{proof}
    Use definition of $\bM$ and Lemma \ref{lem:M-trivial}.
\end{proof}
\paragraph{Dual of upper polynomial polar code}
The dual of an upper polynomial polar code can now be revealed.

\begin{theorem}\label{thm:dual-poly-polar}
     We have
     \begin{equation}
         \C_{\bP\bGN}(\A)^{\bot}=\C_{{P^{t}}^{-1}\bGN\bQ_{\pi}\bGN}(\A^c).
     \end{equation}
\end{theorem}

\begin{proof}
    To determine a polynomial that gives an orthogonal vector over the entire code $\C_{\bP\bGN}(\A)$ we first select indices $i\in\A^c$ and create the canonical vector $\varepsilon_i=(0,\dots,0,1,0,\dots,0)$ having only one $1$ in the position $i.$ Then we compute a vector $\bv$ satisfying $\bv\bM=\varepsilon_i.$ The vector $\bv$ encodes the backward linear combination required for monomial $\bx^{\bin(i)}$ in order to generate a vector which is orthogonal on all codewords of $\C_{\bp\bGN}(\A).$ This can be done since $\bM$ is invertible, and thus we have $\bv=\varepsilon_i\bM^{-1}.$ Also, except for the trivial case when the dimension of the code is $2^m$ and the dual is thus degenerate, the dual always exists since its dimension equals $\A^c.$ This is because all polynomials defining the dual code are linearly independent, i.e, $\{\varepsilon_i\bM^{-1}| i  \in \A^c\}.$ 
    
    To select the polynomials indexed by the frozen set we need to apply our transformation to the monomial basis $\widecheck{\bx^{2^m-1-i}}=\bx^{i}.$ Hence, one polynomial from the dual code is simply  $\varepsilon_i\bM^{-1}\bQ_{\pi}\bGN.$ 

    Use Proposition \ref{pr:M-GP}, and Lemma \ref{lem:prop-perm-psi-G} to obtain the final result. 
\end{proof}

Recall now that when $\bP=\bId$ we have by Lemma \ref{lem:M-trivial} that $\bM=\bGN$, which implies $\C_{\bGN}(\A)^{\bot}=\C_{\bGN\bQ_{\pi}\bGN}(\A^c).$

\begin{corollary}
    We have  \begin{equation}
         \C_{\bP\bGN}(\A)^{\bot}=\C_{{\bP^{t}}^{-1}\bQ_{\pi}\bGN\bQ_{\pi}}(\A^c).
     \end{equation}
\end{corollary}

\begin{remark}
    Since $\bQ_{\pi}\bGN$ is a monomial matrix, and since ${\bP^{t}}^{-1}\bGN$ is an lower triangular invertible matrix, the dual of an upper polynomial polar code is a lower polynomial polar code.  
\end{remark}

% \begin{lemma}
%     The dual of an upper polynomial polar code $\C_{\bP\bGN}(\A)^{\bot}$ admits a decreasing monomial code as sub-code, more exactly, the code $\C_{\bGN}([N-\min(\A),N-1]).$
% \end{lemma}

% \begin{proof}
%     We have $[0,\min(\A)-1]\subset \A^c\Rightarrow \C_{{\bP^{t}}^{-1}\bGN\bQ_{\pi}\bGN}([0,\min(\A)-1])\subset \C_{{\bP^{t}}^{-1}\bGN\bQ_{\pi}\bGN}(\A^c)=\C_{\bP\bGN}(\A)^{\bot}.$
%     Also, since ${\bP^{t}}^{-1}\bGN$ is lower triangular with 1 on the diagonal we have $\C_{{\bP^{t}}^{-1}\bGN\bQ_{\pi}\bGN}([0,\min(\A)-1])=\C_{\bQ_{\pi}\bGN}([0,\min(\A)-1])=\C_{\bGN}([N-\min(\A),N-1]).$
    
% \end{proof}

Let's gather all the results on decreasing monomial sub-codes for polynomial polar and their duals into one final result. 
\begin{proposition}
    Let $\C_{\bP\bGN}(\A)$ be an upper polynomial polar code. Then we have
    \begin{multline}
        \C_{\bGN}([\max(\A^c)+1,N-1])\subseteq \C_{\bP\bGN}(\A)\\\C_{\bP\bGN}(\A)\subseteq \C_{\bGN}([\min(\A),N-1]).
    \end{multline}
        
Moreover,
    \begin{multline}
        \C_{\bGN}([N-1-\min(\A),N-1])\subseteq \C_{\bP\bGN}(\A)^{\bot}\\
        \C_{\bP\bGN}(\A)^{\bot}\subseteq \C_{\bGN}([N-1-\max(\A^c),N-1]).
    \end{multline}
        
\end{proposition}

    \begin{proof}
        We have $\C_{\bGN}([N-\min(\A),N-1])\subseteq \C_{\bP\bGN}^{\bot}\Rightarrow \C_{\bP\bGN}\subseteq \C_{\bGN}([N-\min(\A),N-1])^{\bot}.$ Since $\C_{\bGN}([N-\min(\A),N-1])$ is a decreasing monomial code we have $\C_{\bGN}([N-\min(\A),N-1])^{\bot}=\C_{\bGN}([\min(\A),N-1]).$ 

        The same can be applied to $\C_{\bGN}([\max(\A^c)+1,N-1])\subseteq \C_{\bP\bGN}(\A)\rightarrow \C_{\bP\bGN}(\A)^{\bot} \subseteq \C_{\bGN}([\max(\A^c)+1,N-1])^{\bot}= \C_{\bGN}([N-1-\max(\A^c),N-1]).$
    \end{proof}

\colorlet{O1green}{green!30!orange}
\colorlet{Ogreen}{green!10!orange}
\begin{table}[!ht]
    \centering
    \resizebox{0.48\textwidth}{!}{
    \begin{tabular}{c|c|c|c}
    \toprule
    $N-k$ & $\wm$& Lower bound& $W_{\wm}(\C_{\bGN}(\A))$\\
    \midrule
       \rowcolor{Ogreen}  2 & 2 & 240 & 240 \\

\rowcolor{Ogreen} 3 & 2 & 112 & 112 \\

\rowcolor{Ogreen} 4 & 2 & 48 & 48 \\

\rowcolor{O1green} 5 & 2 & 16 & 16 \\

6 & 4 & 120 & 1240 \\

7 & 4 & 120 & 728 \\

8 & 4 & 120 & 472 \\

9 & 4 & 120 & 344 \\

10 & 4 & 120 & 216 \\

11 & 4 & 120 & 152 \\

12 & 4 & 56 & 88 \\

\rowcolor{Ogreen} 13 & 4 & 56 & 56 \\

\rowcolor{Ogreen} 14 & 4 & 24 & 24 \\

\rowcolor{Ogreen} 15 & 4 & 8 & 8 \\

\rowcolor{O1green} 16 & 4 & 8 & 8 \\

17 & 8 & 28 & 364 \\

18 & 8 & 28 & 236 \\

19 & 8 & 28 & 172 \\

20 & 8 & 28 & 108 \\

21 & 8 & 28 & 76 \\

22 & 8 & 28 & 44 \\

\rowcolor{Ogreen} 23 & 8 & 28 & 28 \\

\rowcolor{Ogreen} 24 & 8 & 12 & 12 \\

\rowcolor{O1green} 25 & 8 & 4 & 4 \\

26 & 16 & 6 & 62 \\

27 & 16 & 6 & 30 \\

28 & 16 & 6 & 14 \\

\rowcolor{Ogreen} 29 & 16 & 6 & 6 \\

\rowcolor{Ogreen} 30 & 16 & 2 & 2 \\

\rowcolor{O1green} 31 & 32 & 1 & 1 \\
\bottomrule
    \end{tabular}
    }
    \caption{Minimum distance, number of minimum weight codewords of the polar code and the sub-code of any generalized polynomial polar code in function of the number of frozen bits for $N=32$}
    \label{tab:lower-bound-min-codewords-32}
\end{table}

\section{Mininum distance and minimum weight codewords}

    Using the algebraic structure of the code, we can also study the minimum distance and the number of minimum‑weight codewords of an upper polynomial polar code. %Another interesting property that we can deduce using the algebraic properties of the code is regarding the minimum distance and the number of such codewords of an upper polynomial polar code. 
    From \cite{li2,rowshan2023minimum,zunker2024enumeration},  the minimum distance of a PAC code is at least that of the underlying plain polar code. The proof  in \cite[Lemma 1]{rowshan2023minimum} holds for any upper polynomial polar code. Here, we make a stronger claim: the minimum distance is preserved. 

    \begin{proposition}
        Let $\C_{\bP\bGN}(\A)$ be an upper polynomial polar code, with $\A$ a decreasing monomial set, and suppose $\C_{\bGN}(\A)$ has minimum distance $2^{m-r}.$ Then, the minimum distance of $\C_{\bP\bGN}(\A)$ is also $2^{m-r}$ and the number of minimum weight codewords is lower bounded by $2^{r}$ for any code rate and length.   
    \end{proposition}

    \begin{proof}
        If $\C_{\bGN}(\A)$ has minimum distance $2^{m-r}$ then $\ev(x_0\cdot x_{r-1})\in \C_{\bGN}(\A)$ due to the fact that $\C_{\bG}(\A)$ is a decreasing monomial code. Indeed, if there is any monomial $f\in\I$ of degree $r$, which evaluates to a minimum weight codeword, then $x_0\cdot x_{r-1}\preceq f.$ Since the row-index corresponding to $x_0\cdots x_{m-1}$ is $2^r+\dots + 2^{m-1}$, the decreasing monomial code $\C_{\bGN}([2^{m}-2^r,2^{m}-1])$ is a sub-code of $\C_{\bGN}(\A).$ We do have that $\C_{\bGN}([2^m-2^r,2^m-1])$ is the code generated by all monomials which are smaller than or equal to $x_0\cdots x_{r-1}$ with respect to $\preceq.$ 
        We thus have $[2^m-2^r,2^m-1]\subset \A$ which implies $\max(\A^c)>2^m-2^r$ and finally $[2^m-2^r,2^m-1]\subset [\max(\A^c)+1,N-1].$ We can infer \[\C_{\bGN}([2^m-2^r,N-1])\subset \C_{\bP\bGN}(\A).\]
        This means that the minimum distance of $\C_{\bP\bGN}$ is lower than or equal to $2^{m-r}.$ Using Lemma 1 from \cite{rowshan2023minimum} we deduce the wanted equality. 
        Also the code $\C_{\bGN}([2^m-2^r,N-1])$ has $2^r$ minimum weight codewords which ends our proof.
        
    \end{proof}

\begin{remark}
    If we take a closer look at the proof of our result we notice that we have a lower bound on the number of minimum weight codewords of an upper polynomial polar code at each code-rate which is stronger than $2^r$, the universal lower bound. Indeed, $W_{\wm}(\C_{\bGN}([2^m-2^r,N-1]))$ can be computed for each code-rate individually. We report these results in Table \ref{tab:lower-bound-min-codewords-32} for $N=32$ and Table \ref{tab:lower-bound-min-codewords-64} for $N=64.$ These quantities were computed using the formula for decreasing monomial codes. 
\end{remark}

\begin{example}
    For the $(64,48)$ polar code the number of minimum weight ($2^2$) codewords is $432=2^{4}+2^{4+1}+2^{4+2},2^{4+3}+2^{4+2}+2^{4+3}$ since we have $x_0x_1x_2x_3,x_0x_1x_2x_4,x_0x_1x_3x_4,x_0x_2x_3x_4, x_0x_1x_2x_5$, $x_0x_1x_3x_5\in\I.$ 

    The number of minimum weight codewords of the PAC code with $\bp=[1,0,1,1]$ equals $320.$ Now consider the decreasing monomial sub-code of the PAC code, namely $\C_{\bGN}([\max(\A^c)+1,N-1])=\C_{\bGN}([34,63]).$ This code has maximum degree monomials $x_0x_1x_2x_3,x_0x_1x_2x_4,x_0x_1x_3x_4,x_0x_2x_3x_4$ which give a total number of minimum weight codewords equal to $240.$ Hence, for this code rate we cannot have less than $240$ codewords of weight 4 (row $N-k=64-48=16$ in Table \ref{tab:lower-bound-min-codewords-64}). 
    %x_0x_1x_2x_5,x_0x_1x_3x_5 polar should be counted differently 2^{1+2} + 2^{1+3}=8+16=24, resulting in 264 (minimum)
    %8+2^7=8+128=132 resulting in 372.
    %16+2^6=16+64=80 resulting in 320.
    
\end{example}

\begin{table}[!ht]
    \centering
    \resizebox{0.4\textwidth}{!}{
    \begin{tabular}{c|c|c|c}
    \toprule
    $N-k$ & $\wm$& Lower bound & $W_{\wm}(\C_{\bGN}(\A))$\\
    \midrule

\rowcolor{Ogreen} 2 & 2 & 992 & 992 \\

\rowcolor{Ogreen} 3 & 2 & 480 & 480 \\

\rowcolor{Ogreen} 4 & 2 & 224 & 224 \\

\rowcolor{Ogreen} 5 & 2 & 96 & 96 \\

\rowcolor{O1green} 6 & 2 & 32 & 32 \\

\rowcolor{O1green} 7 & 2 & 32 & 32 \\

8 & 4 & 496 & 6320 \\

9 & 4 & 496 & 4272 \\

10 & 4 & 496 & 3248 \\

11 & 4 & 496 & 2224 \\

12 & 4 & 496 & 1712 \\

13 & 4 & 496 & 1200 \\

14 & 4 & 496 & 944 \\

15 & 4 & 496 & 688 \\

16 & 4 & 240 & 432 \\

17 & 4 & 240 & 304 \\

18 & 4 & 112 & 176 \\

19 & 4 & 112 & 176 \\

\rowcolor{Ogreen} 20 & 4 & 112 & 112 \\

\rowcolor{Ogreen} 21 & 4 & 48 & 48 \\

\rowcolor{Ogreen} 22 & 4 & 48 & 48 \\

\rowcolor{O1green} 23 & 4 & 16 & 16 \\

\rowcolor{O1green} 24 & 4 & 16 & 16 \\

\rowcolor{O1green} 25 & 4 & 16 & 16 \\

\rowcolor{O1green} 26 & 4 & 16 & 16 \\

27 & 8 & 120 & 2456 \\

28 & 8 & 120 & 1944 \\

29 & 8 & 120 & 1432 \\

30 & 8 & 120 & 1176 \\

31 & 8 & 120 & 920 \\

32 & 8 & 120 & 664 \\

33 & 8 & 120 & 536 \\

34 & 8 & 120 & 408 \\

35 & 8 & 120 & 280 \\

36 & 8 & 120 & 216 \\

37 & 8 & 120 & 152 \\

38 & 8 & 120 & 152 \\

39 & 8 & 56 & 88 \\

\rowcolor{Ogreen} 40 & 8 & 56 & 56 \\

\rowcolor{Ogreen} 41 & 8 & 24 & 24 \\

\rowcolor{Ogreen} 42 & 8 & 24 & 24 \\

\rowcolor{O1green} 43 & 8 & 8 & 8 \\

\rowcolor{O1green} 44 & 8 & 8 & 8 \\

\rowcolor{O1green} 45 & 8 & 8 & 8 \\

46 & 16 & 28 & 556 \\

47 & 16 & 28 & 428 \\

48 & 16 & 28 & 300 \\

49 & 16 & 28 & 236 \\

50 & 16 & 28 & 172 \\

51 & 16 & 28 & 108 \\

52 & 16 & 28 & 76 \\

53 & 16 & 28 & 44 \\

\rowcolor{Ogreen} 54 & 16 & 28 & 28 \\

\rowcolor{Ogreen} 55 & 16 & 12 & 12 \\

\rowcolor{O1green} 56 & 16 & 4 & 4 \\

\rowcolor{O1green} 57 & 16 & 4 & 4 \\

58 & 32 & 6 & 62 \\

59 & 32 & 6 & 30 \\

60 & 32 & 6 & 14 \\

\rowcolor{Ogreen} 61 & 32 & 6 & 6 \\

\rowcolor{O1green} 62 & 32 & 2 & 2 \\

\rowcolor{O1green} 63 & 64 & 1 & 1 \\
\bottomrule
    \end{tabular}
    }
    \caption{Minimum distance, number of minimum weight codewords of the polar code and the sub-code of any generalized polynomial polar code in function of the number of frozen bits for $N=64$}
    \label{tab:lower-bound-min-codewords-64}
\end{table}

\begin{remark}
    For any code length it exist some dimensions $k\in[0,N-1]$ for which no improvements in the number of minimum weight codewords can be achieved, by any upper polynomial polar code.

    Such polar codes are either exactly those containing only one maximum degree monomial, i.e., $\I_r=\{x_0\cdots x_{r-1}\}$ or those containing the maximum degree monomials in the code $\C_{\bGN}([\max(A^c+1),N-1]).$ 
\end{remark}

Let's provide a more detailed analysis of the results reported in Table \ref{tab:lower-bound-min-codewords-64} from bottom to top. We exclude the trivial cases $r=m$ and $r=0.$ 
\begin{itemize}
    \item Max. degree monomials $r=1\Rightarrow x_0$ belongs to the polar code. In this case $\wm=2^{6-1}=32$, and we have $5$ code rates for this case, corresponding to $\I=\{1,x_0\},\I=\{1,x_0,x_1\}, \I=\{1,x_0,x_1,x_2\}, \I=\{1,x_0,x_1,x_2,x_3\}, \I=\{1,x_0,x_1,x_2,x_3,x_4\}.$ 
    The first order Reed-Muller code is not here is because $x_0x_1$ is more reliable than $x_5.$ \\
    Row 62 in Table \ref{tab:lower-bound-min-codewords-64} corresponds to $\I=\{1,x_0\}$ hence the plain polar code is equal to any upper triangular polynomial polar code. The same holds for row 63 which corresponds to $\I=\{1,x_0,x_1\}.$
   Just after this we will no longer have equality since all these polar codes have $x_0x_1$ frozen while $x_2,x_3$ or $x_4$ are in $\I.$ Hence, $x_0x_1$ is a the gap where we notice the first jump.
   \item Max. degree monomials $r=2\Rightarrow x_0x_1$ belongs to the polar code. In this case $\wm=2^{6-2}=16$, and we have $12$ code rates for this case. The first four codes provide equal minimum weight codewords for both plain polar and upper polynomial polar codes. The argument is identical with the case $r=1.$ Indeed, the first four codes will have 
   \begin{itemize}
       \item $x_0x_1\in\I_2\Rightarrow 2^{2}=4$ minimum weight codewords. This is the case for two code rates, namely $\I=\{1,x_0,x_1,x_2,x_3,x_4,x_0x_1\}$ and $\I=\{1,x_0,x_1,x_2,x_3,x_4,x_5,x_0x_1\}.$
       \item $x_0x_1,x_0x_2\in \I_2\Rightarrow 2^{2}+2^{2+1}=12$ minimum weight codewords. Here we have $\I=\{1,x_0,x_1,x_2,x_3,x_4,x_5,x_0x_1,x_0x_2\}.$
       \item $x_0x_1,x_0x_2,x_1x_2\in \I_2\Rightarrow 2^{2}+2^{2+1}+2^{2+2}=28$ minimum weight codewords. Here we have $\I=\{1,x_0,x_1,x_2,x_3,x_4,x_5,x_0x_1,x_0x_2,x_1x_2\}.$
   \end{itemize}
   The gap here is provided by the monomial $x_0x_1x_2$ which is not of degree $2$ but it has a larger index compared to any other degree monomial $x_ix_j$ where at least one of the variables is $x_3\preceq x_i.$ 
   \item Max. degree monomials $r=3\Rightarrow x_0x_1x_2$ belongs to the polar code. In this case $\wm=2^{6-3}=8$, and we have $19$ code rates for this case. The first six polar codes have the same number of minimum weight codewords with the upper polynomial polar codes. 
   \begin{itemize}
       \item $x_0x_1x_2\in\I_3$ gives $2^3=8$ minimum weight codewords. there are 3 codes having only this monomial of degree 3. Up to the next monomial of degree 3 $x_0x_1x_3$ there are 2 monomials in between or equivalently two indices, $x_0x_3,x_1x_3.$ 
       \item $x_0x_1x_3\in\I_3$ gives $2^{3+1}=16$ plus $8$ (from $x_0x_1x_2$) makes $24$ minimum weight codewords. There are only two codes rates with this number of minimum weight codewords, since after $x_0x_1x_3$ the upcoming row corresponds to $x_2x_3.$ 
       \item $x_0x_2x_3$ gives $2^{3+2}=32$ minimum weight codewords, and together with the previous $24$, we obtain $56.$ There is only one polar code with such minimum weight contribution. Here, we notice that we have two code rates for which the lower bound is $56$ while the polar code is one equal to the bound and two strictly bigger. The reason why the second polar code we obtain $88$ minimum weight codewords is because the monomial $x_0x_1x_4$ was considered. This means that there is a frozen pattern between $x_0x_2x_3$ and $x_0x_1x_4.$ Now, why $x_0x_1x_4$ was taken before $x_1x_2x_3$? The reason is because we ordered channels with respect to beta-expansion and we have considered $\beta=2^{1/4}$, which makes $x_0x_1x_4$ more reliable than $x_1x_2x_3.$ 
        After this we notice that the gap arises again from $x_0x_1x_2x_3.$  
   \end{itemize}
\end{itemize}

We compute the code dimensions along with the number ofminimum weight codewords for only those parameters ($N-k$) where is impossible for any upper polynomial polar code to improve the weight distribution of the plain polar code. We report the results for $N=128$ in Table \ref{tab:non-capable-128}.
\begin{table}[!ht]
    \centering
    \resizebox{0.3\textwidth}{!}{
    \begin{tabular}{c|c|c}
    \toprule
    $N-k$ & $\wm$& $W_{\wm}(\C_{\bGN}(\A))$\\
         \midrule
2 & 2 & 4032 \\

3 & 2 & 1984 \\

4 & 2 & 960 \\

5 & 2 & 448 \\

6 & 2 & 192 \\

7 & 2 & 192 \\

8 & 2 & 64 \\

9 & 2 & 64 \\

10 & 2 & 64 \\

11 & 2 & 64 \\

31 & 4 & 224 \\

32 & 4 & 224 \\

33 & 4 & 96 \\

34 & 4 & 96 \\

35 & 4 & 96 \\

36 & 4 & 96 \\

37 & 4 & 96 \\

38 & 4 & 32 \\

39 & 4 & 32 \\

40 & 4 & 32 \\

41 & 4 & 32 \\

42 & 4 & 32 \\

43 & 4 & 32 \\

70 & 8 & 112 \\

71 & 8 & 112 \\

72 & 8 & 48 \\

73 & 8 & 48 \\

74 & 8 & 48 \\

76 & 8 & 16 \\

77 & 8 & 16 \\

78 & 8 & 16 \\

79 & 8 & 16 \\

80 & 8 & 16 \\

81 & 8 & 16 \\

101 & 16 & 56 \\

102 & 16 & 24 \\

103 & 16 & 24 \\

104 & 16 & 24 \\

105 & 16 & 8 \\

106 & 16 & 8 \\

107 & 16 & 8 \\

108 & 16 & 8 \\

118 & 32 & 28 \\

119 & 32 & 12 \\

120 & 32 & 4 \\

121 & 32 & 4 \\

125 & 64 & 6 \\

126 & 64 & 2 \\

127 & 128 & 1\\
\bottomrule
    \end{tabular}
}
    \caption{Co-dimensions for impossible improvement with upper polynomial polar codes for $N=128.$}
    \label{tab:non-capable-128}
\end{table}

Notice that co-dimension for where impossible improvements are revealed usually come in consecutive sequences, e.g., for $N=128$ we have $N-k$ in $[1,11]\cup[31,43]\cup([70,81]\setminus\{75\})\cup[101,108]\cup[118,121]\cup[125,127].$ We have computed this list for $N=256$ as well and we have obtained
$[1,16]\cup[49,73]\cup[123,128]\cup[136,147]\cup[188,204]\cup[228,235]\cup[248,249]\cup[253,255].$
\def\ranks{0, 1, 2, 6, 3, 8, 9, 18, 4, 10, 11, 21, 13, 23, 25, 37, 5, 12, 14, 24, 16, 27, 29, 41, 19, 30, 32, 43, 35, 46, 48, 56, 7, 15, 17, 28, 20, 31, 33, 44, 22, 34, 36, 47, 39, 49, 51, 58, 26, 38, 40, 50, 42, 52, 53, 59, 45, 54, 55, 60, 57, 61, 62, 63}

\begin{figure*}
    \centering
\begin{tikzpicture}[scale=0.25]

% Column headers (binary, vertical)
\foreach \rank [count=\col from 0] in \ranks {
    \pgfmathtruncatemacro{\bzero}{int(mod(\col/32,2))}
    \pgfmathtruncatemacro{\bone}{int(mod(\col/16,2))}
    \pgfmathtruncatemacro{\btwo}{int(mod(\col/8,2))}
    \pgfmathtruncatemacro{\bthree}{int(mod(\col/4,2))}
    \pgfmathtruncatemacro{\bfour}{int(mod(\col/2,2))}
    \pgfmathtruncatemacro{\bfive}{int(mod(\col,2))}

    \node[rotate=90, anchor=south] at (\col+1.3,1)
        {\ttfamily \scriptsize\bzero\bone\btwo\bthree\bfour\bfive};
}

% Heatmap
\foreach \row in {1,...,64} {
    \foreach \rank [count=\col from 0] in \ranks {
        \ifnum\rank<\row
            \def\cellcolor{white}
        \else
            \def\cellcolor{red}
        \fi
        \fill[\cellcolor] (\col,-\row) rectangle ++(1,-1);
        \draw[black] (\col,-\row) rectangle ++(1,-1);
    }
}

% Row labels
\foreach \r in {1,...,64} {
    \node[anchor=east] at (-0.4,-\r-0.5) {\ttfamily \r};
}

\end{tikzpicture}
    \caption{Polar code profile for dimension $k\in\{1,\dots,N\}$ (represented on rows) for $N=64.$ Red boxes are frozen indices ($\A^c$) while white are information indices ($\A$). Computed using $\beta$-expansion with $\beta=2^{1/4}.$}
    \label{fig:polar_profile}
\end{figure*}
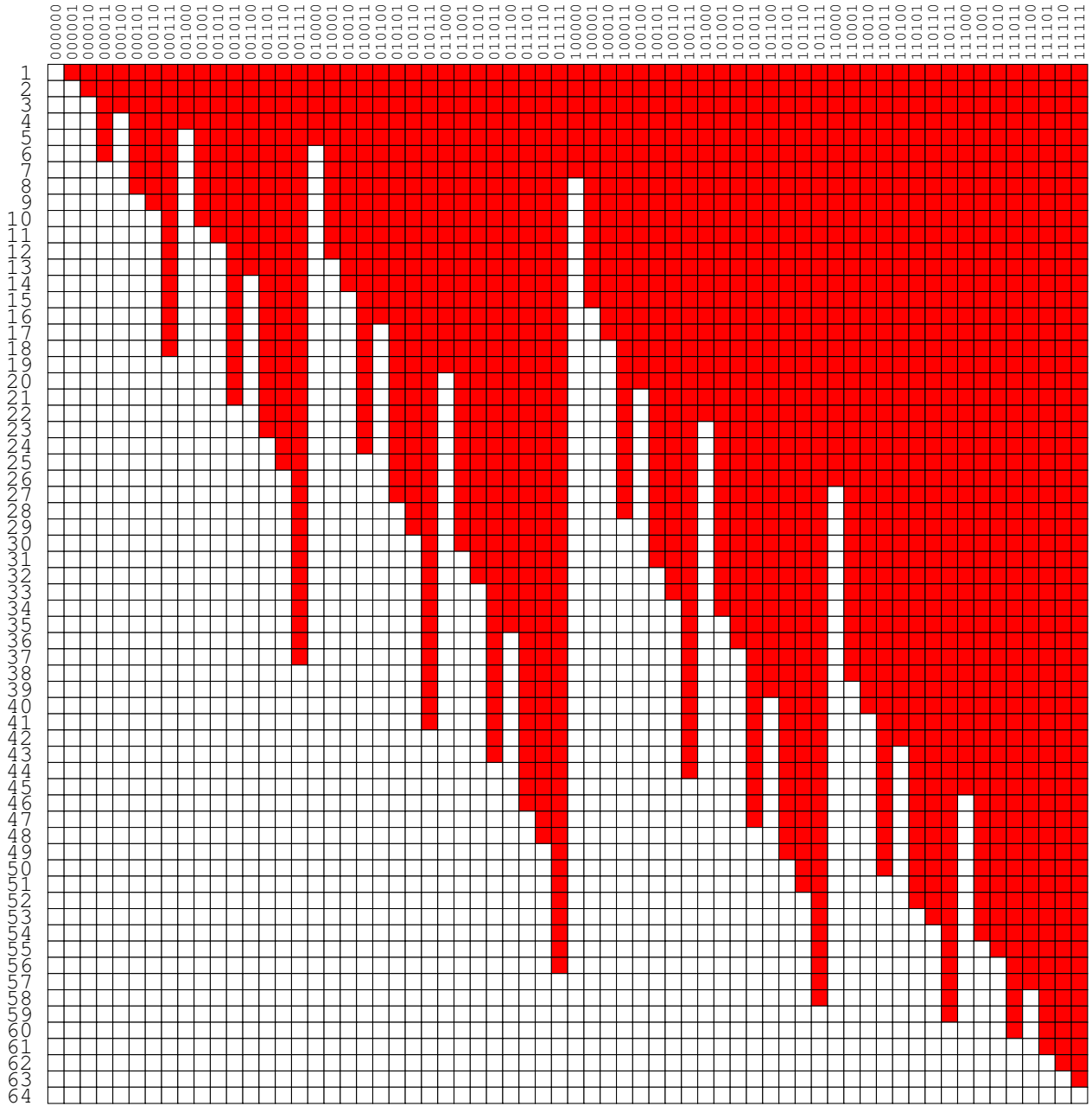

% \newcommand{\N}{128}
% \def\ranks{0,1,2,6,3,8,9,19,4,10,11,23,13,25,27,40,5,12,14,24,15,26,28,41,17,29,30,43,32,45,47,63,7,16,20,31,21,33,34,48,22,35,36,49,37,51,52,60,38,53,54,61,55,64,65,74,42,56,57,66,58,67,69,75,44,59,62,68,70,76,77,82,46,71,72,78,73,79,80,83,50,81,84,85,86,88,89,93,39,87,90,94,91,95,96,98,92,97,99,100,101,103,104,107,104,102,105,108,106,109,110,111,112,113,114,115,116,117,118,119}

% \begin{tikzpicture}[scale=0.12]

% \foreach \rank [count=\col from 0] in \ranks {
%     \pgfmathtruncatemacro{\bzero}{int(mod(\col/64,2))}
%     \pgfmathtruncatemacro{\bone}{int(mod(\col/32,2))}
%     \pgfmathtruncatemacro{\btwo}{int(mod(\col/16,2))}
%     \pgfmathtruncatemacro{\bthree}{int(mod(\col/8,2))}
%     \pgfmathtruncatemacro{\bfour}{int(mod(\col/4,2))}
%     \pgfmathtruncatemacro{\bfive}{int(mod(\col/2,2))}
%     \pgfmathtruncatemacro{\bsix}{int(mod(\col,2))}

%     \node[rotate=90, anchor=south] at (\col+0.5,0.3)
%         {\ttfamily \bzero\bone\btwo\bthree\bfour\bfive\bsix};
% }

% \foreach \row in {1,...,\N} {
%     \foreach \rank [count=\col from 0] in \ranks {
%         \ifnum\rank<\row \def\cellcolor{white} \else \def\cellcolor{red} \fi
%         \fill[\cellcolor] (\col,-\row) rectangle ++(1,-1);
%         \draw[black] (\col,-\row) rectangle ++(1,-1);
%     }
% }

% \foreach \r in {1,...,\N} {
%     \node[anchor=east] at (-0.4,-\r-0.5) {\ttfamily \r};
% }

% \end{tikzpicture}

\printbibliography

@article{Arikan,
  author  = {Erdal Ar{\i}kan},
  title   = {Channel Polarization: A Method for Constructing Capacity-Achieving Codes for Symmetric Binary-Input Memoryless Channels},
  journal = {IEEE Transactions on Information Theory},
  volume  = {55},
  number  = {7},
  pages   = {3051--3073},
  year    = {2009},
  doi     = {10.1109/TIT.2009.2021379}
}

@article{liu2024method,
  title={A Method to Reduce the Complexity of Computing the Complete Weight Distribution of Polar Codes},
  author={Liu, Zhichao and Ma, Zhiming and Yan, Guiying},
  journal={arXiv preprint arXiv:2410.17872},
  year={2024}
}

@techreport{arikan2,
author = {E. Ar\i kan},
title = {From sequential decoding to channel polarization and back again},
type = {preprint},
year = {2019},
archivePrefix = {arXiv},
eprint = {1908.09594},
}

@techreport{li2,
author = {B. Li and H. Zhang and J. Gu},
title = {On Pre-transformed Polar Codes},
type = {preprint},
year = {2019},
archivePrefix = {arXiv},
eprint = {1912.06359},
}

@article{rowshan-precoding,
author = {M. Rowshan and E. Viterbo},
journal = {IEEE Globecom (GC Wkshps), Madrid, Spain},
pages = {1-6},
title = {On Convolutional Precoding in PAC Codes},
year = {2021},
}

@incollection{bardet2016crypt,
author = {M. Bardet and J. Chaulet and V. Dragoi and A. Otmani and J.-P. Tillich},
booktitle = {Post-Quantum Cryptography, vol. 9606},
pages = {118-143},
title = {Cryptanalysis of the {McEliece} Public Key Cryptosystem Based on Polar Codes},
year = {2016},
}

@inproceedings{bardet,
author = {M. Bardet and V. Dragoi and A. Otmani and J.-P. Tillich},
booktitle = {IEEE Int. Symp. Inf. Theory (ISIT)},
pages = {230-234},
title = {Algebraic properties of polar codes from a new polynomial formalism},
year = {2016},
}

@phdthesis{dragoi17thesis,
author = {Vlad Florin Dragoi},
title = {An algebraic approach for the resolution of algorithmic problems raised by cryptography and coding theory},
year = {2017},
school={Normandie Universit{\'e}}
}

@inproceedings{yao,
author={Yao, Hanwen and Fazeli, Arman and Vardy, Alexander},
  booktitle={2021 IEEE Int. Symp. Inf. Theory (ISIT)}, 
  title={A Deterministic Algorithm for Computing the Weight Distribution of Polar Codes}, 
  year={2021},
  volume={},
  number={},
  pages={1218-1223},
  }

@article{rowshan2023minimum,
  title={On the minimum weight codewords of PAC codes: The impact of pre-transformation},
  author={Rowshan, Mohammad and Yuan, Jinhong},
  journal={IEEE Journal on Selected Areas in Information Theory},
    year={2023},
  volume={4},
  number={},
  pages={487-498}
}

@INPROCEEDINGS{rowshan22EnuPAC,
author = {M. Rowshan and J. Yuan},
journal = {IEEE Information Theory Workshop (ITW), Mumbai, India},
pages = {255-260},
title = {Fast Enumeration of Minimum Weight Codewords of {PAC} Codes},
year = {2022},
}

@article{rowshan-pac1,
author = {M. Rowshan and A. Burg and E. Viterbo},
journal = {in IEEE Trans. on Vehicular Technology},
month = {February},
number = {2},
pages = {1434-1447},
title = {Polarization-adjusted Convolutional ({PAC}) Codes: Sequential Decoding vs List Decoding},
volume = {70},
year = {2021},
}

@INPROCEEDINGS{rowshan2024weight,
  author={Rowshan, Mohammad and Drăgoi, Vlad–Florin and Yuan, Jinhong},
  booktitle={2024 IEEE Int. Symp. on Inf. Theory (ISIT)}, 
  title={Weight Structure of Low/High-Rate Polar Codes and Its Applications}, 
  year={2024},
  volume={},
  number={},
  pages={2945-2950},
}

@article{PBL21,
author = {C. Pillet and V. Bioglio and I. Land},
journal = {IEEE Information Theory Workshop (ITW},
month = {October},
pages = {1-6},
title = {Polar codes for automorphism ensemble decoding},
volume = {2021},
year = {2021},
}

@ARTICLE{mondelli-construction,
  author={Mondelli, Marco and Hassani, S. Hamed and Urbanke, Rüdiger L.},
  journal={IEEE Trans. on Inf. Theory}, 
  title={Construction of Polar Codes With Sublinear Complexity}, 
  year={2019},
  volume={65},
  number={5},
  pages={2782-2791}
}

@inproceedings{zunker2024enumeration,
  title={Enumeration of Minimum Weight Codewords of Pre-Transformed Polar Codes by Tree Intersection},
  author={Zunker, Andreas and Geiselhart, Marvin and Ten Brink, Stephan},
  booktitle={58th Annual Conference on Information Sciences and Systems (CISS)},
  pages={1--6},
  year={2024},
  organization={IEEE}
}

@ARTICLE{vlad1.5d,
  author={Drăgoi, Vlad-Florin and Rowshan, Mohammad and Yuan, Jinhong},
  journal={IEEE Trans. Commun.}, 
  title={On the Closed-Form Weight Enumeration of Polar Codes: 1.5d -Weight Codewords}, 
  year={2024},
  volume={72},
  number={10},
  pages={5972-5987},
}

@INPROCEEDINGS{Wang-Dragoi2023,
  author={Wang, Hsin-Po and Dragoi, Vlad-Florin},
  booktitle={2023 IEEE Int. Symp. Inf. Theory (ISIT)}, 
  title={Fast Methods for Ranking Synthetic BECs}, 
  year={2023},
  volume={},
  number={},
  pages={1562-1567},
}

@article{gu2025pac,
  title={PAC Codes Meet CRC-Polar Codes},
  author={Gu, Xinyi and Rowshan, Mohammad and Yuan, Jinhong},
  journal={arXiv preprint arXiv:2501.18080},
  year={2025}
}

@inproceedings{gu2024reverse,
  title={Reverse PAC Codes: Look-ahead List Decoding},
  author={Gu, Xinyi and Rowshan, Mohammad and Yuan, Jinhong},
  booktitle={2024 IEEE Int. Symp. Inf. Theory (ISIT)},
  pages={2844--2849},
  year={2024},
  organization={IEEE}
}

@inproceedings{gu2023rate,
  title={Rate-compatible shortened pac codes},
  author={Gu, Xinyi and Rowshan, Mohammad and Yuan, Jinhong},
  booktitle={2023 IEEE/CIC International Conference on Communications in China (ICCC Workshops)},
  pages={1--6},
  year={2023},
  organization={IEEE}
}

@InProceedings{dragoi-2019,
author="Dr{\u{a}}goi, Vlad
and Beiu, Valeriu
and Bucerzan, Dominic",
editor="Lanet, Jean-Louis
and Toma, Cristian",
title="Vulnerabilities of the McEliece Variants Based on Polar Codes",
booktitle="Innovative Security Solutions for Information Technology and Communications",
year="2019",
publisher="Springer International Publishing",
address="Cham",
pages="376--390"
}

@article{Arikan2020_PAC_Systematic,
  author       = {Erdal Ar{\i}kan},
  title        = {{Systematic Encoding and Shortening of PAC Codes}},
  journal      = {Entropy},
  year         = {2020},
  volume       = {22},
  number       = {11},
  pages        = {1301},
  doi          = {10.3390/e22111301},
  url          = {https://www.mdpi.com/1099-4300/22/11/1301}
}

@article{Moradi2024_PAC_Cyclic,
  author       = {Mohsen Moradi},
  title        = {{Polarization-Adjusted Convolutional (PAC) Codes as a Concatenation of Inner Cyclic and Outer Polar- and Reed–Muller-like Codes}},
  journal      = {Finite Fields and Their Applications},
  year         = {2024},
  volume       = {93},
  pages        = {102321},
  doi          = {10.1016/j.ffa.2023.102321},
  url          = {https://doi.org/10.1016/j.ffa.2023.102321}
}

\end{document}